\documentclass[EJP]{ejpecp} % replace ECP by EJP if needed.
\usepackage{array}    %% AS                                                              %%
\SHORTTITLE{Symmetric Locks, Bombs and Testing Model}

\TITLE{Symmetric Case of Locks, Bombs and Testing Model}

\AUTHORS{  Isaac M. Sonin\footnote{University of NC at Charlotte.}
 \footnote{ \EMAIL{imsonin@uncc.edu}}}
 % Separate items with ;

\AMSSUBJ{91A27, 90B40} % Edit. Separate items with ;
\SUBMITTED{July 5, 2023} % Edit.
\ACCEPTED{x, y} % Edit.

\VOLUME{0}
\YEAR{2023}
\PAPERNUM{0}
\DOI{10.1214/YY-TN}

\ABSTRACT{We present a Defense/Attack resource allocation model, where Defender has some number of ``locks"\ to protect $n$ vulnerable boxes (sites), and Attacker is trying to destroy these boxes, having $m$ ``bombs"\ that can be placed into the boxes. Similar models were studied in game theory - (Colonel) Blotto games, but our model has a feature absent in the previous literature. Attacker tests the vulnerability of all sites before allocating her resources, and these tests are not perfect, i.e., a test can be positive for a box without a lock and negative for a box with a lock. We describe the optimal strategies for a special case of a general Locks-Bombs-Testing (LBT) model when all boxes are identical and the Defender has a fixed number of locks.}

\newcommand{\lcn}{allocation }

\newcommand{\crsl}{correspondingly }
\newcommand{\dst}{distribution }

\newcommand{\dsts}{distributions }

\newcommand{\flw}{following }
\newcommand{\fls}{follows}

\newcommand{\lgt}{algorithm}

\newcommand{\clcs}{calculations}

\newcommand{\crsp}{corresponding}
\newcommand{\ndp}{independent}
\newcommand{\ndpc}{independence }
\newcommand{\mpt}{important }
\newcommand{\imd}{immediately }

\newcommand{\prb}{probability }
\newcommand{\prbe}{probability}

\newcommand{\prbs}{probabilities }

\newcommand{\stg}{strategy }
\newcommand{\stge}{strategy}
\newcommand{\stgs}{strategies }

\begin{document}

\section{Introduction. General and Symmetrical LBT Models. Main Results}\label{INT}
%%%%%%%%%%%
\footnotetext{This paper was written with active participation of Konstantin I. Sonin, (Univ. of Chicago Harris School of Public Policy) but due some difference of opinion on the structure of the paper, he withdrew his signature.\par
%%%%%%%%%%%%%% Liu Li ?? Fedor nbx send + Mazalov  Mark Whitmeyer ??  Xingjie Li,
We thank Michael Grabchak, Ernst Presman, Fedor Sandomirsky and Alexander Slastnikov, for their valuable remarks, helpful discussions, and patience with reading numerous drafts.}
%%%%%%%%%%%%%%
The problem of allocating limited resources between different tasks is a classical problem in many areas of Operations Research, Economics, Finance, and Engineering. This problem, in the context of players (participants), is an important field in Game Theory. In a classical Colonel Blotto game, two players distribute limited resources between different sites (battlefields) with the goal of winning more sites, where you win a site if you have more resources on this site than your opponent. There is substantial literature on this topic, starting with the classic paper (Shubik and Weber, 1981) \cite{shub81} and more recent publications, such as (Robertson, 2006) \cite{rob06}, where a complete solution of the ``continuous" version was given, as well as (Hart, 2015) \cite{Har15}, where some interesting extensions are discussed. In a comprehensive and detailed survey (Hohzaki, 2016) \cite{hoh16} dedicated to Search Games, the Blotto game is classified as an attack-defense game. Another survey is (Hausken and Levitin, 2012) \cite{HaLe12}. There are even more papers dedicated to these games and, as in all of Operations Research, all classifications have many overlapping parts. As an example of other important papers on attack-defense games we mention (Golany et all, 2012) \cite{Rot12} and (Powell, 2007) \cite{pow07}, and recent paper (Clarkson et all, 2023) \cite{Cla23} on Search Game in discrete locations. \par
%%%%%%%%%%%%%%%%% +add Hausen Levitin 1 or 2 Bib line
The inspiration for the model in this paper and some basic ideas can be traced to the paper by (K. Sonin and A. Wright, 2017, 2023) \cite{sw19}, where they provided a model of intelligence gathering in combat and used highly detailed data about Afghan rebel attacks, insurgent-led spy networks, and counterinsurgency operations. This theoretical model was a novel version of the Colonel Blotto game. In their model, first, the government allocates its scarce defense resources across possible targets. Then, each target is independently
tested by the rebels for vulnerability. Finally, the rebels base their choice of targets on the results of these tests. Empirically, the paper demonstrated a robust link between local economic conditions and the patterns of rebel attacks. \par
%%%%%%%%%%%% \vspace{1cm} where one ?? important special case was solved.
A more general and abstract mathematical model called the Locks, Bombs and Testing (LBT) model was described in
(I. Sonin, 2019) \cite{son19} and \cite{son22}. The solution for the other important case was obtained in the PhD thesis of Liu Li in 2019. This thesis, in a modified form, is a substantial part of the paper (Li and Sonin, 2021) \cite{lison21}. \par
%%%%%%%%%% \cite{li19}.  \cite{Li19}
%%%%%%%%%%% Symmetrical LBT model , where all boxes are identical. %a \emph{vector of values} (costs) $c=(c_i,i=1,...,n)$.
The mathematical tools used in this paper, where all parameters are discrete, are rather elementary, but our model touches the key notions in three important areas: Probability Theory, Statistics, and Game Theory. These are:
%Finally, if one will try to describe Probability Theory, Statistics and Game Theory, each in no more than three words, then they probably will sound like this:                     ??? Mojet ubrat next lines ??
in Probability Theory - independence, total probability formula, Bayes' formula; in Statistics - test of hypotheses, errors of type 1 and 2, or in equivalent terms - sensitivity, specificity; in Game Theory - randomization and (Nash, von Neuman) equilibrium points. The model presented here combines all of these ingredients.\par
%\vspace{.1cm}
The main goal of our paper is to give a complete solution to a specific case of a general model, so called
Symmetrical-LBT, (S-LBT) but we prefer to describe first the general model and an approach to its solution and only after to describe the specifics for our case. Thus, in the next subsection we describe the General LBT model (G-LBT) and a scheme to solve it and in the subsection after we show how this scheme works for S-LBT. This material and Section 2 is almost verbatim text from (Sonin, 2022) \cite{son22}.
%%%%%%%%%%%%%%%%%% = fr SS22-Sept
\subsection{Setup for General LBT }\label{setup}
%%%%%%%%%%% 1.2 %%%%%%%%%%% == Sept 22 no tam \subsection{Setup}\label{setup} %%%%%%%%
General LBT (G-LBT) model can be described as \fls. As in most attack-defense games, there are two players with different roles. We call one of them the Defender (DF) and the other, the Attacker (AT). They are ``fighting" over $n$ ``boxes"\ (sites, battlefields, cells, targets, time slots, etc.) with possibly distinct cost values $c_i$ for box $i$. ``Boxes" and ``sites" are used  interchangeably in this paper. AT is trying to destroy these boxes by placing ``bombs"\ that can result in explosions and destruction. One or more bombs can be placed into the same box. AT has $m$ identical bombs to allocate among $n$ boxes. The bombs are imperfect. We denote by $p$ the probability of explosion of any one bomb and $p(u)$ the probability of explosion of at least one of $u$ bombs placed into the same box. A box is destroyed if at least one explosion occurs, and the explosions of different bombs in the same site or in different sites are independent.   \par
%%%%%%%%%%% We refer hereafter to $p(u)$ as the \emph{explosion function}.
DF is trying to protect the boxes by distributing among them identical ``locks". A lock is a protection device which, when placed in a box, prevents its destruction with any number of bombs in it. We assume that locks are perfect but this assumption can be weakened in a more general model. Obviously, ``locks" and ``bombs" are just the names for discrete units of resources of protection and destruction. The number of locks $k$, $k<n$, can be fixed, in which case it is an $A\equiv A(n,k,m)$ problem, the subject of paper \cite{son22}, and this paper, or it can be a random variable obtained when a lock appears in site $i$ with probability $\lambda_i$ independently of other boxes, in which case we call this a
$B\equiv B(n,\lambda, m)$ problem. The latter assumption can represent either the uncertainty of DF about the number of locks,
%%             resources that will be available to her
or the uncertainty of AT about how many locks will be distributed. With $m=1$ we will use simpler notation $A(n,k)$ and $B(n,\lambda)$. When all sites are identical we call this model a symmetrical LBT model. The symmetrical $B(n,\lambda)$ model was solved in \cite{lison21}. \par
%% The case B was analyzed in \cite{lison21}.
The crucial feature of both models, in contrast to classical Blotto games, is that AT can and will \emph{test} every box, trying to find boxes without locks. The result of testing is a vector of signals $s=(s_1,...,s_n)$, where each $s_i=0,1$, is known only to AT. Hereafter we refer to this vector as \emph{signal} $s$.
% and $\sum_{i}\lambda_i=k$, . The former case is the subject of Here we discuss only the case that number of locks is a
The following notation is used throughout the paper. Define random variables $T_{i},S_{i},D_{i},$ each taking two values $0$ and $1$: $T_{i}=1$ when the $i$th box is protected, and $0$ otherwise; $S_{i}=1$ when the $i$th box test is positive (plus), and $0$ otherwise (minus), and $D_{i}=1$ when the $i$th box is destroyed. For simplicity of notation, the absence of subindex $i$ means that the formula applies to any box.\par
%\vspace{.1cm}
%%%%%%%%%%%%	%%%%%%%%%%% Thus $|\Gamma^{k}|=M_k$. Hereafter we use $M=M_k$ where confusion is impossible.
The testing is not perfect: a test of site $i$ may have a positive result ($S_i=1$) even if there is no lock at the site ($T_i=0$), and negative ($S_i=0$) even if there is a lock ($T_i=1$). The probabilities of correct identification of both types, in statistical language the \emph{sensitivity} and \emph{specificity}, $P(S_i=1|T_i=1)=a_i$ and $P(S_i=0|T_i=0)=b_i$, the values (costs) of all boxes $c_i, i=1,...,n$, the number of locks $k$, and the number of bombs $m$, are known to both players. Our assumptions imply the following basic equalities:
	%%%%%%%%%%% &=& P(S_i=1|L_i=1)=a, P(S_i=0|L_i=0)=b,
	\begin{eqnarray}
		 P(S_i=1|T_i=1, S_{-i}=s_{-i})&=&a_i, \ \ \ \ P(S_i=0|T_i=0, S_{-i}=s_{-i})=b_i, \notag \\
		 P(D_{i} =1|T_{i}=1)&=&0,\ \ \ \ \ \ P(D_{i} = 1|T_{i}=0,u_{})=p(u_{}),\label{be}
	\end{eqnarray}
where $S_{-i}$ is vector $S=(S_1,...,S_n)$ without coordinate $S_i$, the similar notation is used for vector $(s_{-i})$, $u=u_{i}$ is the number of bombs in box $i$, and $p(u) $ is the \emph{explosion function}, the probability of at least one explosion in a box with $u$ bombs.\par
%%%%%%%%%%%%%%%%%%%%%%%%%
%\vspace{.1cm}
If success is independent across bombs, then the function $p(u)=1-\left(1-p\right) ^{u}$. Thus, the function $p(u)$ is increasing and upward concave, and the function $\Delta p(u)\equiv p(u+1)-p(u)$ is decreasing. The diminishing effect of each extra bomb, i.e., the diminishing of its marginal utility, plays an important role in determining the optimal strategy. Most of our results can be extended to a more general explosion function $p(u)$ so long as it is positive, increasing, and upward concave. %%%%%% , but for simplicity in this paper we assume that $m=1$ and $p=1$.%%%
\par
%\vspace{.1cm}
% Note that we denote the values of rv $S_i$ as $s_i$ to use the notation $s_1,...,s_N, N=2^n$ to list \vspace{.1cm}
 \begin{remark}[]\label{Rem1}
   Note that there is a fundamental difference between models $A\equiv A(n,k,m)$ and $B\equiv B(n,\lambda,m)$. In model
$A$ the signals $S_1,...,S_{n}$ are not independent, though the result of testing in box $i$ depends only on whether a lock is present or absent in this box. This property is reflected in the first line of formula (\ref{be}). In model $B$ the signals $S_1,...,S_{n}$ are independent.
 \end{remark}
%\vspace{.1cm}
% Note that we denote the values of rv $S_i$ as $s_i$ to use the notation $s_1,...,s_N, N=2^n$ to list all possible
If e.g. in model $A(3,1)$ testing in boxes 1 and 2 is very accurate and indicates that they do not have a lock, one knows that a lock must be in box 3. In model $B(3, \lambda)$ the results of testing in boxes 1 and 2 give no information about the presence of a lock in box 3. Only its testing gives some information. Thus in model $B$ boxes are informationally independent. This makes model $B$ in some regards simpler than $A$ although both have different specific features. In this paper we consider only case  $A(n,k,m)$ though from time to time we will make short remarks about case $B$.
%%%  Case $B(n,\lambda,m)$ was treated in \cite{lison21}. \par
%%%%%%%%%%%%% The case $k=0$ means that no locks are allocated. The collection of $\pi_k(\gamma ), k=0,1,2,...,n$ defines the strategy of DF, hereafter $\pi(\gamma )$. \par
\subsection{Goals and Strategies}\label{goals}
An \mpt parameter of the G-LBT model is the cost vector $c=(c_1,...,c_n)$, where WLOG we assume that
$c_1\geq c_2\geq...\geq c_n=1$. The goal of AT is to maximize the \emph{total expected value} of destroyed boxes, and the goal of DF is to minimize it. Thus, we obtain a zero-sum game. Note that if players have different goals, e.g., if they have different cost vectors, or if AT is trying to destroy at least some fixed number of sites, and the goal of DF is still to minimize the total expected value of destroyed boxes, then we would obtain a non-zero-sum game.  \par     %%% $c_{AT}, c_{DF}$
%%% or exactly three sites, but DF does not know which. note immediately
%%%%%%%%%  $K$ parameter WLOG, we can assume that the allocation of bombs is deterministic and
%%%   in problem A, or vector $\lambda_i,i=1,...,n$ in problem B...and their actual number in problem B are
When the number of available locks is $k, 1\leq k\leq n$, a \emph{strategy} of DF is a probability distribution $\pi_{}$, (prior $\pi)$, on a set $\Gamma_{n,k}\equiv \Gamma$ of all possible locations of $k$ locks.  We assume that set $\Gamma$ is an \emph{indexed} set of $n$-dim vectors $\gamma =(r_{1},r_{2},...,r_{n})$, where each $r_i$ equals zero or one, $\sum_{i=1}^nr_i=k$ and $r_i=1$ means that box $i$ has a lock.  For simplicity of notation, we write $i\in\gamma$ to mean that $r_i\equiv r_i(\gamma)=1$ and $i\notin\gamma$ to mean that $r_i\equiv r_i(\gamma)=0$. Note that $|\Gamma_{n,k}|= \binom{n}{k}=:M_{n,k}\equiv M$. Thus \stg $\pi_{}$ is an $M$-dim vector $(\pi(\gamma), \gamma \in \Gamma)$. We also will use notation $\pi =(\pi_{1},...,\pi_{M})$. Let $G$ be a $M\times n$ matrix, where all $\gamma\in\Gamma$ serve as rows. Thus, entries with ones (zeros) in matrix $G$ show where locks are present (absent). All vectors, written in the text are \emph{row vectors}, in all matrix calculations they are \emph{column vectors}. Vector with coordinates $q_i$ will sometimes be denoted as $(q_i)$. \par
%Matrix $A^T$ denotes the transpose of $A$.   , including the fixed number of locks $k$. \vspace{.1cm}
%%%%%%%% Bayesian
In our setting we assume that both DF and AT know all parameters of the model. As we will see later, as usual in game theory, whether the prior distribution of locks $\pi$, selected by DF, is known to AT, or not, does not matter. DF does not know the results of testing by AT, i.e. the signal $s$.\par
%% a prior OR the prior , but for simplicity of presentation we will assume that $\pi$ is known, but not the positions of the locks.We assume that both DF and AT know all  parameters of the model. specifying the distribution of the random number of locks, %\par %%%%%%%%% \par \vspace{.1cm}
After the locks are allocated using prior $\pi$, AT tests all boxes, receives signal $s=(s_1,...,s_n)$, and then, using $\pi$ and the probabilities of signals $s$ given a fixed position of locks $\gamma$, $p(s|\gamma)$, calculates the \emph{posterior distribution} of the positions of locks $\theta(\gamma|s,\pi)$.
%At this stage, the number of available bombs $m$ and explosion function $p(u)$ do not participate. some rand-n or a
After that, for each signal $s$ and each $m$, and knowing $\pi$, AT solves the problem of optimal allocation of $m$ bombs trying to maximize the total expected value of destroyed sites. Such allocation can be deterministic or use a randomization when there is an indifference in the choice of boxes. An \emph{optimal strategy of AT} $\varphi_{opt}(\pi)$, with respect to the strategy of DF $\pi$, \emph{is a collection} of her optimal responses
$\varphi_{opt}(s|\pi)=(u_1(s|\pi),...,u_n(s|\pi))$ to each signal $s$, where $u_i(s|\pi)$ is the number of bombs placed into site $i, i=1,...,n$, $\sum_iu_i(s)=m$. The prior distribution $\pi$ together with this optimal strategy $\varphi_{opt}(\pi)$ results in the corresponding \emph{total expected damage} to DF, $v(\pi)=v(\pi, \varphi_{opt}(\pi))$. \par
%%%%%%%%%%%%  and the number of bombs $m$. \vspace{.1cm}
DF does not know signals $s$ but can reproduce all calculations of AT and thus knows $\varphi_{opt}(s|\pi)$ and $v(\pi)$. The goal of DF is to select a prior distribution of locks $\pi_*$, which minimizes this damage. Then the pair
$(\pi_*, \varphi_*)$, $\varphi_*=\varphi_{opt}(\pi_*)$, forms a classical Nash equilibrium (NE) point. The corresponding game value is denoted $v_*=v(\pi_*, \varphi_*)$. Although $\pi_*$, as a rule, is not unique,
%all these $\pi$ have some common property that results in a unique AT strategy,
the minimization of $v(\pi)\equiv v(\pi, \varphi_{opt}(\pi))$ over all $\pi$ gives a specific game value $v_*$.\par
%%%%%%%%%% We denote these games with equal values of $c_i$ as
The main drawback of a general G-LBT model is that the volume of calculations is growing exponentially because e.g. the number of possible signals is $2^n$. More realistic models with polynomial volume of calculations is Symmetrical (S-LBT) model, the main goal of this paper.
\subsection{Symmetrical LBT}\label{sym}
%%%%%%%%%% ss 1.4 p 9 c_0
When, in the general model described above, all sites are symmetrical, i.e., all costs are the same, $c_i=1$, and all parameters of testing are the same, $a_i=a, b_i=b$, (in Problem $B$ $\lambda_i=\lambda$ for all $i$), we call this game the \emph{symmetrical} LBT (S-LBT) game (model). The additional justification to limit the consideration to the symmetrical case in this paper and in \cite{lison21} is the following. The general LBT game and the straightforward approach to solving it, described above, have two basic drawbacks. First, the set of possible positions for locks, i.e., the set of subsets of an $n$ element set, generally have order $4^n/\sqrt n$, and the set of potential signals has size $2^n$. As a result, the calculations of posterior distributions $\theta(\gamma|s,\pi)$ and their marginal distributions $\alpha_i(s|\pi)=P(T_i=0|s, \pi)$, which play a crucial role in the description of optimal strategies, become cumbersome for large $n$. The second problem is that the knowledge of detailed information about the values of $k, c_i, a_i$, and $b_i$ in many practical situations is unrealistic.\par
%cases  the main focus in \cite{lison21} was on the analysis of a simpler S-LBT model, where all sites have identical values $c_i=1$, and all $a_i=a, b_i=b$. , and in the Appendix we outline the proof.
The symmetry of sites and testing implies that the optimal \stg of DF is simple: to allocate all $k$ locks at random between $n$ sites. Intuitively, this statement is clear and it can be proved formally that deviation from this allocation will reduce the optimal payoff in equilibrium point. A reader can assume that this is just as an assumption in this model. Thus, the game is reduced to a problem of maximization for AT, given uniform prior distribution of locks, $\pi=\pi_*$. Since in S-LBT $\pi_*$ is known, in many formulas we skip further reference to $\pi_*$.\par
%%%%%%%%%%%% that the \prb that a particular box has a lock is $t=k/n$ %%%% it was a TYPO In a file sent toA OR: was slip instead of skip % To analyze ??  $r_A(x)\equiv r(x)$ in problem $A(n,k,m)$ we need a more formal description of the random experiment.\par
In statistical physics, the distribution of $k$ available objects (locks) between $n$ boxes uniformly, with no more than one object in a box, is called the Fermi-Dirac statistics: any combination of $k$ protected boxes has the same probability $1/\binom{n}{k}$. The symmetry implies that the probability of protection, i.e. the presence of a lock for each individual box, $P(T_i=1)=\frac{k}{n}$. We formally prove these claims in Lemma \ref{Lem1}. \par
%%%%%%% \beta_i(s)=1-\alpha_i(s)$ is $t The S-LBT model with independent locks, model $B(n,\lambda)$, was completely solved in \cite{lison21}.
The main goal of our paper is to present the complete solution of the $A(n,k)$ S-LBT game, but we also provide, for comparison, some details  about the completely solved model $B\equiv B(n,\lambda)$ from {\cite{lison21}. \par
%% We present some numerical results in Section \ref{example}.\par
%%%%%%%  and \cite{sons20} $A(n,k)$ $B(n,\lambda)$ \par 	 %%%%%%%%%%%%%%%%%\newpage \subsection{From G-LBT to S-LBT }\label{from}
\vspace{.1cm}          %%%%%%%%%%%
%%%%%%%%%%%%%%% =\frac{\alpha_i(s)}{P(T=0|S=1,x)}
We will see in the general scheme of solution of G-LBT, described in the next Section, that an \mpt role is played by the terms $\alpha_i(s |\pi)=P(T_i=0|s,\pi)$, the \prb of absence of a lock in box $i$, given signal $s$ and prior $\pi$. The key to the full solution of S-LBT is the fact that now $\alpha_i(s |\pi_*)\equiv P(T_i=0|s, \pi_*)=P(T_{i}=0|s_{i}, N=x)$, where $N$ is the number of minuses (zeros) in signal $s$. In other words, the exact positions of boxes with zeros do not matter, only their number. This leads to one of our main results about S-LBT - the optimal strategy of AT $\varphi_{opt}(\cdot |s,m)$, given signal $s$, depends only on the probability of an explosion $p$, the value $x$ of rv  $N$, and the following ratio $r_A(x)\equiv r_{n,k}(x)$.
%% % in S-LBT model it is sufficient only to analyze two critical ratios %%, where we use shorthand notation $P(\cdot |N=x)=P(\cdot |x)$, \vspace{.1cm}
\begin{eqnarray}
r_A(x)=\frac{P(T_i=0|S_i=0,N=x)}{P(T_i=0|S_i=1,N=x)}\equiv\frac{p^{-}(x)}{p^{+}(x)},\ \ 0<x<n, \label{rxr}
\end{eqnarray}
where $p^{-}(x)=P(T=0|S=0,N=x), p^{+}(x)=P(T=0|S=1,N=x)$. More than that, in Lemma \ref{Lem2} we will obtain the explicit formulas for $p^{-}(x)$ and $p^{+}(x)$ through parameters $n,k,b$ and distributions $g_{n-1,k}$ and $g_{n,k}$ - distributions of rv $N$ in models $A(n-1,k)$ and $A(n,k)$.\par
%%%%%%%%%%%%%%%%%%
In Problem $B$, a similar role is played by ratio $r_B=\frac{P(T_i=0|S_i=0)}{P(T_i=0|S_i=1)}$. Note that this ratio does not depend on $x$. Corresponding formula is given in Subsection 6.2.
% These statements will be proved in the key Lemma 1 in Section Y. The similar is true for problem $B$, where the \crsp \ ratio $r_B$ depends on the parameters of the model, $n,k,\lambda, a,b$.\par %%%%%%%%%
An interesting and even counterintuitive property, proved later, is that in problem $A$, the function $r_A(x)$ and
therefore the optimal \stg and the value function, depend on $a$ and $b$ only through the value
%%%% ubral \ds $\ds c=\frac{a}{1-a}\frac{b}{1-b}, 0\leq c <\infty$, a combined i dalee
$c=\frac{a}{1-a}\frac{b}{1-b}, 0\leq c <\infty$, a combined characteristic of the quality of testing. It means in particular that WLOG we can assume that $a=b=\frac{\sqrt{c}}{1+\sqrt{c}}$. In problem $B$ this property is not true with respect to the value $r_B$. When the parameters of sensitivity $a$ and specificity $b$ are ``informative" (``truthful"),
i.e.  $a\geq \frac{1}{2}, b\geq \frac{1}{2}, a+b>1$, these ratios are more than one. This immediately implies that if the number of bombs does not exceed the number of minuses (zeros) $N=x$ in signal $s$, then each bomb goes to some  minus box, and this allocation is arbitrary with respect to the positions of minus boxes.
% will have a simpler structure than in the general case, symmetrical with respect to all sites with minus (and correspondingly plus) signals,
When, given signal $s$ with $N=x, 0\leq x\leq n$, the number of bombs $m$ exceeds $x$, an optimal strategy of AT $\varphi_{opt}(\pi_*)$ is again arbitrary with respect to the positions of minus boxes, and can be expressed through $r_A(x)$, and other parameters, as follows. \par
%Given that the number of bombs available is $m$, given $N=x,$ $0\leq x\leq n$, strategy $\varphi$ has the following structure.
Initially, all bombs are placed one by one into each of $x$ minus boxes until the threshold level $d=d_A(x)$ in
Problem $A$ (level $d=d_B$ in Problem $B$) is reached in each of them or the bombs are exhausted. Afterwards, the bombs are added one by one to plus boxes until there is a bomb in each of them. Then, bombs are added one by one into minus boxes until each of these boxes has $d+1$ bombs in each of them, then back to plus boxes until each has at least $2$ bombs, etc. This \textquotedblleft fill and switch"\ process stops when AT runs out of bombs. We will call such a strategy a $d$\emph{-uniform as possible} strategy (hereafter, a ``$d$-UAP strategy"). If $x=0$ or $n$, then all boxes are simply filled sequentially, and this is a $0$-UAP strategy. The value of the threshold $d_A(x)$, given that $x$ minuses were observed, represents the ``advantage level"\ of a minus box over a plus box. A similar interpretation can be given to the threshold $d_B$. The thresholds $d_A(x)$ will be described in Theorem 1. The values $N=x$ and $r_A(x)$ ($r_B(\lambda)$) play the role of sufficient statistics in the optimization problem. \par
%%%%%%%%%%%%%%
In an example with $n=5$, $x=3$, $m=12$, and $d(x)=3$, the dynamics of $3$-UAP strategy described above implies that finally  $3$ bombs are placed into each of $2$ minus boxes, $4$ into the third minus box, and one bomb into each of the $2$ plus boxes. If $d(x)=2$, than with the same $n,x,m$ each of the minus boxes has 3 bombs, $1$ plus box has $1$, and the second plus box has 2.\par
%%%%%%%%%%%% As we assumed that the success is independent across attacks, $p(u)=1-\left(1-p\right) ^{u}$. The function $p(u)$ is increasing and upward concave, and the function $\Delta p(u)\equiv p(u+1)-p(u)$ is decreasing. The diminishing effect of each extra attack will play an important role in determining the optimal strategy. %% \par %%% Sec 2   %%%%%%	%%%%%%%%%%%%  We describe these results briefly in Appendix. \vspace{.1cm}
%% so called \emph{Universal} having special properties like ``bombs"\ daet levye kfvyxhki !!! to the next Section,
Before proceeding we mention that the G-LBT model is very flexible and can be easily extended in many directions. We discuss several possible interpretations and extensions in Section \ref{GEN}.
%%%%%%%%%%%% The main goal of our paper is to present a full solution for some special unsymmetrical cases G-LBT model,
The structure of our paper is as follows. In Section \ref{INT} we briefly described the main results of our paper. In Section \ref{STG} we describe the General Scheme to solve G-LBT. In Section \ref{PREM} this scheme is transformed and adjusted to S-LBT model. In Section \ref{OSV} we formulate and prove Theorem \ref{th:1}. In Section \ref{GEN} we discuss some open problems and Section \ref{APN} - Appendix contains some technical proofs from statements in previous Sections. % and in Section ?? \ref{lisonin_example} we consider some extra examples. in OSV present some preliminary formulas and auxiliary results. \par %%%%%%%%%% 1.2 %%%%%%%%%%%% \subsection{Goals and Strategies and }\label{goals}
\section{Three Stages to Solve G-LBT. General Scheme  }\label{STG} %% Sec 2
%%%%%%%%%%% Nbx shorter !!!!!!!!!!! %%%%%% subsec 2.2  we keep notation $s_1,...,s_N$ for the list of all possible signals $s$. %% This Section is a shorter version of \crsp \ parts in paper \cite{son22}.\par
The distribution of locks and subsequent testing may be viewed as a two-stage random experiment with outcomes represented by pairs $(\gamma ,s)$, where $\gamma$ and $s$ are $n$-dimensional vectors. The probability of each outcome $(\gamma ,s)$ is $p(\gamma ,s)=\pi(\gamma )p(s|\gamma )$, where $\pi$ is the prior distribution of locks, and
$p(s|\gamma )=P(S_{1}=s_{1},...,S_{n}=s_{n}|\gamma )$. The explicit formula for $p(s|\gamma)$ is given in Proposition \ref{Prop1}.\par
%%%(G,s) are points in Omega - prb spave that we prefer not to introduce !!!!!!!!!!!
%%%%%%%%%%%%%%%% $G$ be a $M\times n$ matrix   (s,\gamma,\pi) $A(s,\pi)$,  and $D(s,\pi)$,$\Gamma\equiv\Gamma_{k}$
%\vspace{.1cm}
%%%%%%%%%%%%% s=(s_{1},...,s_{n})$ is a (vector) \emph{signal} about boxes vulnerability. The set of all such vectors is
The solution of the G-LBT game, i.e., the description of all NE points, can be divided conventionally into three stages.
%%%% [stages]; we have to solve three problems. $G$ be a $M\times n$ matrix consists of three steps
\par
\vspace{.1cm}
\textbf{Stage 1. Prior Distribution, Testing, and Posterior Distributions}\par
%%%%% =(s_1,...,s_n) %%%%%%%% \pi_0  of all possible positions of locks , \gamma \in \Gamma) , $\sum_i\pi_i=1$
\vspace{.1cm}
Suppose that DF selected prior \dst $\pi=(\pi(\gamma))$, i.e. $M$-dim \prb vector
$\pi=(\pi_1,...,\pi_M)$. If $\pi$ is concentrated on a \emph{particular} vector $\gamma=(r_{1},r_{2},...,r_{n})$, then testing of AT generates random vector, signal $s=(s_1,...,s_n)$, with formula for the \prb of this vector
$p(s|\gamma)=\prod_{i=1}^{n}P(S_i=s_i|T_i=r_i)$, defined by vectors of sensitivity and specificity
$a=(a_1,...,a_n)$ and $b=(b_1,...,b_n)$. We denote $\Sigma$ to be the \emph{indexed} set $(s(1),s(2),...,s(N))$ of all signals $s$, $|\Sigma|=2^n=N$. We consider $p(s|\gamma)$ as the entries of the $M\times N$ stochastic matrix $P$.\par
%%%%%%%%%%%%% $P$, , i.e. both DF and AT, \equiv P(s|\gamma)  , equation (\ref{psg})
%\vspace{.1cm} %% matrix \Theta(s,\gamma,\pi)  , s\in \Sigma      (Post.Dist.matrix) $\Theta(\pi)$   in matrix notation $(p(s|\pi))=P^T\pi$,
Given $\pi$, one can calculate $N$-dim vector $(p(s|\pi))$, with entries $p(s|\pi)=\sum_\gamma\pi(\gamma)p(s|\gamma)$ for each signal $s$, and then, using Bayes' formula, the $N\times M$ stochastic matrix of \emph{posterior \dsts of locks} $\Theta(\pi)$ with entries $\theta(\gamma|s,\pi)$. The formula for the entries of this matrix is given in Proposition \ref{Prop1}. \par %%%%%%%%%%%%
Given $\pi$, one can obtain $n$-dim vector $\alpha(\pi)$ with entries
$\alpha_i(\pi)=P(T_i=0|\pi)$ - the \prb of \emph{absence}  of a lock in box $i$. Thus $\alpha_i(\pi)=\sum_{\gamma\in \Gamma}\pi(\gamma)P(T_i=0|\gamma)$, where $P(T_i=0|\gamma)$ equals $1$, if $r_i=0$ and $0$ if $r_i=1$ in vector $\gamma=(r_{1},r_{2},...,r_{n})$. Given vector $\gamma$, we also use the shorthand notation $i\in \gamma$ if $r_i=1$, and
$i\notin \gamma$ if $r_i=0$. Then, given $\pi$, we can equivalently express,
$\alpha_i(\pi)=\sum_{\gamma: i\notin \gamma}\pi(\gamma)$. Note, that, given $k$ locks, for any $\pi$ we have $\sum_{i=1}^n\alpha_i(\pi)=n-k$, since the complimentary sum of \prbs of \emph{presence} of $k$ locks in $n$ boxes is $k$.\par
%%%%% In matrix notation, using matrix $G$, we have $\beta(\pi)=G\pi$, and $\alpha(\pi)=G^c\pi $.\par
%We remind that, given $k$ locks, for any $\pi$ we have $\sum_{i=1}^n\beta_i=k$ and $\sum_{i=1}^n\alpha_i=n-k$.  \par
%%%%%%%%%%%%%% $B\equiv B(\pi)=\{\beta_i(s |\pi)\}$, and $\sum_{i=1}^n\beta_i=k$ and
%\vspace{.1cm}
Similarly, given prior $\pi$ and signals $s$, using posterior \dst matrix $\Theta(\pi)=\{\theta(\gamma|s,\pi)\}$, one can obtain an $N\times n$ matrix $A(\pi)=\{\alpha_i(s,\pi)\}$, where $\alpha_i(s,\pi)=P(T_i=0|s,\pi)$, the \prb of absence of a lock in site $i$, given prior $\pi$ and signal $s$. \par
\vspace{.1cm}
% \begin{proposition}[My proposition]    Body of the proposition.  \end{proposition}
%%% former Rem it rmvd
\begin{remark}\label{Rem2}%%% [My proposition]    Body of the proposition.  \end{proposition}
At Stage 1, the number of bombs $m$, the explosion function $p(u)$, and cost vector $c$ do not participate. In a sense this stage may be considered as a statistical problem where a Statistician (AT), given signal $s$, deals with the problem of recognizing the pattern of allocation of $k$ objects among $n$ sites, assuming knowledge of a prior \dst of objects placed by Nature (DF).
\end{remark}
%\vspace{.1cm} %%%%%%%%%% \varphi  Given  a \stg of DF $\pi$,%% In matrix notation, $A(\pi)=G^c\Theta(\pi)$. \par
%%%%%%%%%%%%%% $B(\pi)=G\Theta(\pi),\vspace{.2cm}
\vspace{.1cm}
%%%%%%%%%%%% i.e., the optimal allocation of bombs for each signal $s$
\textbf{Stage 2. Given $\pi$, selected by DF, both AT and DF calculate an optimal \stg of AT $\varphi_{opt}(\pi)$ and
the expected loss (damage) of DF $v(\pi)\equiv v(\pi,\varphi_{opt}(\pi))$}.\par
\vspace{.2cm}
%%%%%%%%%%%%% dimensional \equiv\beta=(\beta_i), i=1,...,n Given.. this vectorgiven prior $\pi$ and signal $s$,
Once matrix $A(\pi)$ is obtained, to find the optimal response of AT given signal $s$, with $m=1$ one has to compare the expected \emph{potential} loss $l_i(s,\pi)=pc_iP(T_i=0|s,\pi)\equiv pc_i\alpha_i(s,\pi)$, in site $i$ if a bomb is placed there, over all $i$. These entries form an  $N\times n$ matrix of\emph{ potential losses} $L(\pi)$ with rows ($l_i(s,\pi))$.
%%%%   Diag=Dig%% In matrix notation $L(\pi)=Diag(c)A(\pi)$, where $c=(c_1,...,c_n)$ is  the $n$-dim cost vector. \par
%%%%%%%%%%% $D(\pi)$ $\Phi_{opt}(\pi)$
If in a row defined by signal $s$ there is only one maximal element then AT should place her bomb there. If the maximum is achieved in $t$ boxes, we assume that AT selects one of them at random. Formally, this optimal choice can be represented by a $N\times n$ matrix $\Phi_{opt}(\pi)=\{u_{i}(s,\pi)\}$ with $u_{i}(s,\pi)=1$, if entry
$l_i(s,\pi)$ is the only maximal element in row $s$ and equal to $1/t$ if the maximum is achieved in $t$ entries. All other entries $u_{i}(s,\pi)=0$. This matrix describes the \emph{optimal \stg }$\varphi_{opt}(\pi)$ of AT, with respect to the \stg of DF $\pi$. If $m>1$, for each signal $s$, bombs can be placed in a few boxes, and with more than one bomb in some boxes. Then matrix $\Phi_{opt}(\pi)$ will represent the number of bombs in each box, but the optimal allocation will be defined by a more complex procedure taking into account explosion function $p(u)$. \par
%%%%%%%%%%%%%% $n$-dimensional vector . The column $s$ of matrix $\Phi_{opt}(\pi)$ describes an optimal allocation of a bomb given signal $s$, and thus, %%%%%%%%
After matrix $L(\pi)$ is calculated and hence $\varphi_{opt}(\pi)$ is known, one can obtain sequentially: $N\times n$ matrix of\emph{ optimal} (for AT) damages to DF, $D_{opt}(\pi)=\{d_i(s,\pi,\varphi_{opt}(\pi))\}$, where
$d_i(s,\pi,\varphi_{opt}(\pi))=l_i(s,\pi)u_{i}(s |\pi)$, the expected  damage to DF in site $i, i=1,...,n$ when AT applies her optimal \stg given signal $s$;  $n$-dimensional vector $d(\pi)\equiv d_{opt}(\pi,\varphi_{opt}(\pi))$ with coordinates
$d_i(\pi)\equiv d_i(\pi,\varphi_{opt}(\pi))$, the (optimal) expected optimal damage to site $i$, averaged over all signals, i.e. $d_i(\pi)=\sum_{s\in \Sigma}p(s|\pi)d_i(s,\pi,\varphi_{opt}(\pi))$, and, finally, (optimal) \emph{total expected damage} to DF $v(\pi)\equiv v(\pi,\varphi_{opt}(\pi))=\sum_{i=1}^nd_i(\pi)$. \par
%%%%%%%%%%%%%%%%%%%%%%%%%%%%\
\vspace{.1cm}
%%%%%%%%%%
Both DF and AT can produce all of these calculations and learn an optimal \stg of AT for each signal $s$, although \emph{only AT knows the signals}. Thus, they both know matrix $D_{opt}(\pi)$, vector $d(\pi)$ and value $v(\pi)$ for each $\pi$. Before discussing Stage 3, where DF selects $\pi_*$, minimizing this value, we make some comments about the first two stages. \par
\vspace{.1cm}
On the theoretical side, the first two stages follow rather standard applications of the total probability formula and Bayes' theorem. However, in practical terms, to find an optimal response for each signal is a challenging problem, especially since the number of signals grows exponentially. The complete solution of the symmetrical LBT is partly possible because, in this
case, DF uses a uniform prior, and all signals $s=(s_1,...,s_n)$ with the same number $x$ of $s_i=0$, i.e., the number of negative (zeroes) results of testing, $x=n-\sum_is_i$, contain the same amount of information. \par
\vspace{.1cm} %% %\vspace{.2cm}
\textbf{Stage 3. DF, knowing $v(\pi,\varphi_{opt}(\pi))$ for each $\pi$, finds all $\pi_*$ that minimize this quantity.} \par
\vspace{.1cm}
Each $\pi_*$, minimizing expected loss of DF, generates Nash equilibrium (NE) point $(\pi_*,\varphi_{opt}(\pi_*))$. Indeed,
DF found the \emph{minmax} solution and we have a matrix zero-sum game with unique game value, hence $(\pi_*,\varphi_*)$ is NE point. Thus for each player it is not profitable to change her strategy if the other player does not change hers. All such NE points have the same \emph{game value} $v_*=min_\pi v(\pi_*,\varphi_{opt}(\pi_*))$. Note that the optimal response of AT for each signal $s$ in Stage 2 is based not on the posterior \dst matrix $\Theta(\pi)$ itself, but on the matrix of marginal posterior \prb $A(\pi)$.
As a rule, there are many prior $\pi$ that can produce the same $A(\pi)$, and therefore they will produce the same \stg of AT $\varphi_{opt}(\pi)$, but we also can not exclude the possibility of priors $\pi_1$ and $\pi_2$, that produce different $A_j(\pi)$, and hence potentially different $\varphi_{opt}(\pi_j), j=1,2$, but both provide the same minimal value $v_*=v(\pi_1,\varphi_{opt}(\pi_1))=v(\pi_2,\varphi_{opt}(\pi_2))$. Such situation does not occur in three specific models solved in paper \cite{son22}. In the general case the existence of multiple $A(\pi)$ is an open problem. This is not important for the zero-sum game in this paper but it is important for non-zero-sum modifications of G-LBT model.  \par
% was If this may occur then one of possible solutions would be that DF will reveal to AT her optimal \stg $\pi_*$. Then AT will have no choice as to apply $\varphi_{opt}(\pi_*)$. %There are examples that show that such situation is possible.\\
%The optimal \stgs $(\pi_*,\varphi_{opt}(\pi_*))$ and game value $v_*$ depend on all parameters of the model, in particular on the number of locks $k$, but we do not indicate these relationships in out notation. \par %%%%%%%%%%%%% , described above,
\vspace{.1cm}
Stage 3 is the most difficult, (``da liegt der Hund begraben"). Paper \cite{son22} among other results gives the full solution of a G-LBT game with $n=2,k=1$, cost vector $(c_0,1), c_0\geq 1$ and general informative testing parameters $a,b: 1/2\leq a,b \leq 1 $. This solution shows how nontrivial Stages 2 and 3 are. This solution also shows that a potential hypothesis that in a NE point all sites have the same expected
losses, is, generally, not true. Strangely enough it is true for some interval of values of the parameter $c_0$. Stage 3 for $n>2$ probably does not have a universal solution. There is no explicit solution even for problem $A(3,1)$, where prior $\pi$ is only two dimensional, but every $\varphi_{opt}(\pi)$ depends additionally on eight $3$-dimensional signals $s$.
The hope is to develop an \lgt \ of recursive approximations.
%, starting from some \stge, e.g., a solution to a similar problem when testing is not informative.
%% We give the examples of ad hoc solutions in Examples C1 and C2 in Subsection \ref{numer}. We discuss also
A possible general approach for Stage 3 in described in \cite{son22}. \par
\vspace{.1cm}
%%%%%%%%Flow chart mvd to AFTER $(\pi_*,\varphi_{opt}(\pi_*))$.
To summarize a few statements given above, we present Proposition \ref{Prop1}, which holds for a general model - G-LBT.
%This Proposition was obtained in \cite{son22}. \par
The description of matrix $P$ with rows $p(s|\gamma)$, is given in point a), where the position of locks $\gamma$ is fixed, and thus does not depend on a prior \dst of locks. Then, given prior $\pi$, we can obtain vector $p(s|\pi)=P(S_1=s_1,...,S_n=s_n|\pi)$. We remind that the definition of parameters $a_i$ and $b_i$ is given in formula (\ref{be}), and $i\in \gamma$ means that $r_i(\gamma)=1$, and $i\notin \gamma$, means that $r_i(\gamma)=0$. \par
%were given in the first lines of  Stage 2. \par
%%%%%%%%% Matrix $P(s|\gamma)=\{p(s|\gamma)\}\equiv P$ is $N\times M$ matrix.\par
%%%%%%%%%%% formula (\ref{gtx})%\newpage   $P(T=\gamma|S=s)=\pi(\gamma|s)$.s, \gamma       position gam
\vspace{.1cm} %% \newpage
\begin{proposition}[G-LBT model]\label{Prop1}
  a) Given $\gamma \in \Gamma$ and signal $s=(s_1,...,s_n)$, $p(s|\gamma)$ is given by the formula %(\ref{ptx}), i.e.}
%\vspace{-.5cm}
\begin{eqnarray}
p(s|\gamma)=\prod_{i=1}^{n}p(s_i|\gamma)=\prod_{i\in \gamma}a_i^{s_i}(1-a_i)^{1-s_i}\prod_{i\notin \gamma}b_i^{1-s_i}(1-b_i)^{s_i}. \label{psg}
\end{eqnarray}
%%%%%%%%%%\varphi $(p(s))\equiv (p(s|\pi_*))$
b) Given prior $\pi$ and signal $s$, the coordinates of vector $p(s|\pi)$ are given by the formula
\begin{eqnarray}
p(s|\pi)=\sum_{\gamma\in \Gamma}\pi(\gamma)p(s|\gamma).       \label{pspi}
%%%%%%%%%%%%%%%%%% \ \ \ \ \ \ \ (p(s|\pi))=P\pi.
\end{eqnarray}
%%%%%%%%%%%%%%
c) Given prior $\pi$ and signal $s$, for any $\gamma$, posterior \prb $\theta(\gamma|s,\pi)$ is given by the formula
%\vspace{-.2cm}
%%%%%%%%%%% \Theta=\{\theta(\gamma|s,\pi)\};\ \ \ \
\begin{eqnarray}
\theta(\gamma|s,\pi)=\frac{P(\gamma,s|\pi)}{p(s|\pi_{})}=\frac{\pi_{}(\gamma)p(s|\gamma)}{p(s|\pi_{})}\equiv\frac{\pi(\gamma)p(s|\gamma)}
{ \sum_{\sigma\in \Gamma}\pi(\sigma)p(s|\sigma)}. \label{pgs}
%%%%%%%%%%% , \ \ \ \Theta=Dig(\pi)P^TDig(1/p(s|\pi))
\end{eqnarray}
d) Given prior $\pi_{}$  and signal $s$, the entries of matrix $A(\pi)=\{\alpha_i(s,\pi)\}$ are given by formula
%of matrix $B(\pi)$ and vector $\beta(\pi)$ are given by}$\alpha_i(\pi)=\sum_{\gamma: i\in A_0(\gamma)}\pi(\gamma)$.
%\vspace{-.2cm}
\begin{eqnarray}
\alpha_i(s,\pi)=\sum_{\gamma: i\notin \gamma}\theta(\gamma|s,\pi),\ \ \ \  \  \sum_{i=1}^n\alpha_i(s,\pi)=n-k. \label{alp}
\end{eqnarray}
%%% , \ \ \ \A(\pi)=G_c\Theta(\pi)
 \end{proposition}
 %%%%%%%%%%%%%%% \textbf{Proposition 1, G-LBT model.}
%%%In matrix notation PDL is given by $N\times M$ matrix $P(\gamma|s,\pi)=\{p(\gamma|s,\pi)\}.$ It can be presented as...???
%%%%%%%%%%%%%% \frac{p(t,x)\binom{n}{x}}{g(x)\binom{n}{k}}=s(t|x)b(t,x), for all $s$ and $\gamma$,
%%%%%%%%%%%%%%%%%%%%%%%%%%% \vspace{-.2cm} %%%%%$A_0(\gamma)$ \par \vspace{.1cm}
%%% Note that these formulas for S-LBT have much simpler form.\par \vspace{.2cm} ...of Proposition 1
%\begin{proof}  \end{proof}
\begin{proof} a) Given $\gamma$, the signals in different sites are \ndp, and hence
$p(s|\gamma)=\prod_{i=1}^{n}p(s_i|\gamma)=\prod_{i\in \gamma}p(s_i|\gamma)\prod_{i\notin \gamma}p(s_i|\gamma)$.
Signal $S_i$ in $i$-th box is a Bernoulli rv with parameter $p=a_i$ if $i\in \gamma$, i.e. $r_i(\gamma)=1$, and $p=1-b_i$ if $i\notin \gamma$, i.e. $r_i(\gamma)=0$. The pmf of a Bernoulli rv has the form $p(x)=p^x(1-p)^{1-x}, x=0,1$, and then
$\prod_{i\in \gamma}p(s_i|\gamma)=\prod_{i\in \gamma}a_i^{s_i}(1-a_i)^{1-s_i}$. Similarly,
$\prod_{i\notin \gamma}p(s_i|\gamma) =\prod_{i\notin \gamma}b_i^{1-s_i}(1-b_i)^{s_i}$. Thus, we obtain formula (\ref{psg}).
Note, that matrix $P=\{p(s|\gamma), s\in \Sigma, \gamma\in \Gamma\}$ is a stochastic one since the sum of entries in every $\gamma$ row is one. \par
Formula (\ref{pspi}) in b) is an application of Total Probability formula.
%may be considered as a scalar product or the result of matrix multiplication. \par
%%%%%%%%
The equality in formula (\ref{pgs}) represents Bayes' formula. For each signal $s$ we also have obvious equality $\sum_{\gamma}\theta(\gamma|s,\pi)=1$. The first equality in (\ref{alp}) is an analog of the
equality $\alpha_i(\pi)=\sum_{\gamma: i\notin \gamma}\pi(\gamma)$ presented in the description of Stage 2, with
prior $\pi$ replaced by posterior \prb $\theta(\gamma|s,\pi)$. The second equality is similar to the
equality $\sum_{i=1}^n\alpha_i(\pi)=n-k$, also mentioned earlier.
%%%%%%%%%%% Added 5/31
The formal proof is based on the equality\par
\vspace{.05cm}
$p(s|\pi)\sum_{i=1}^n\alpha_i(s,\pi)=\sum_{i=1}^n\sum_{\gamma: i\notin \gamma}\pi(\gamma)p(s|\gamma)=\sum_{i=1}^n\sum_{\gamma}(1-r_i(\gamma))\pi(\gamma)p(s|\gamma)$. \par
%%%%%%%%%%
Using the equalities
$\sum_{\gamma}\pi(\gamma)p(s|\gamma)=p(s|\pi)$, and $\sum_{i=1}^nr_i(\gamma)=k$, we obtain that $\sum_{i=1}^n\alpha_i(s,\pi)=n-k$.
\end{proof}
%%% Note that now both equalities in (\ref{alp}) hold for each signal $s$.%\vspace{.2cm} %%%%%%  We describe these results briefly in Appendix. \vspace{.1cm}
\section{Preliminary Formulas and Auxiliary Results for S-LBT}\label{PREM} 	%%%%%%%%%%%%%%%%% $A(n,k)$ problem,
In this section we transform the general formulas for matrices and vectors from G-LBT to specific formulas for S-LBT where DF uses \emph{uniform} $\pi_*$, the testing parameters and cost values are the same for all sites, ($c_i=1$ for all $i$). They will lend themselves to simple and easily programmable computational formulas. We remind that $N$ is the number of minuses (zeros) in signal $s$. %\par %\vspace{.1cm}
\subsection{Lemma 3.3}\label{Lemma1} %% former L1
%%%%%%% given that the number of locks $k$ is fixed,
In statistical physics, the distribution of $k$ available objects between $n$ boxes uniformly, with no more than one object in a box, is called the Fermi-Dirac statistics: any combination of $k$ filled boxes has the same probability $1/\binom{n}{k}$. The symmetry of this \dst implies also the following:\par
\vspace{.2cm}
%that the probability of the presence of an object for each individual box, $P(T_i=1)=\frac{k}{n}$. \par
%%%%% We formally prove these claims in Lemma 1 ??.fr Sym Sec 4 We skip the step is to show that, for optimization
% \begin{proposition}[My proposition]    Body of the proposition.  \end{proposition}
\begin{proposition}[] \label{Prop4}   a) Given prior $\pi_*$, %(\ref{ptx}), i.e.}
%\vspace{-.3cm}
\begin{eqnarray}
P(T_i=1)=k/n, \ \ \ \ \ P(T_i=0)=(n-k)/n. \label{pti}
\end{eqnarray}
b) Given prior $\pi_*$, for any signal $s=(s_1,...,s_n)$ and any $x=0,1,...,n,$
%%%%%%%%%%%%%% useful \vspace{-.2cm}
\begin{eqnarray}
P(S=s|N=x)=1/\binom{n}{x}, \ \ \ \ \ P(S=s,N=x)=P(N=x)/\binom{n}{x}. \label{psx}
\end{eqnarray}
\end{proposition}
%\textbf{Proposition 4.}\vspace{.1cm} %%%%%%%%%%%%%%%% EDIT %\begin{proof}  \end{proof}
\begin{proof} a) The number of combinations of $k$ locks having one lock on a fixed position and the other $k-1$ locks having any of remaining $n-1$ positions is $\binom{n-1}{k-1}$. Then, given prior $\pi_*$,  $P(T_i=1)=\binom{n-1}{k-1}/\binom{n}{k}=\frac{k}{n}$. \par
%%%%%%%%%%%
\vspace{.1cm}
b)  This point is an intuitively appealing observation, that, given prior $\pi_*$, signals distribution is uniform conditional on the number $x$ of minus signals. Indeed, by symmetry of allocation of locks and symmetry of testing, all signals with the same number of zeros (minuses) $N=x$ have the same \prbe, and since the number of such signals is $\binom{n}{x}$, then we obtain two equivalent equalities in formula (\ref{psx}). 
\end{proof}
%\par%%% the aposterior \dst of locks (  $b(\gamma|s)$ %%%%% \vspace{.1cm}
The key to the drastic reduction in calculations in S-LBT is a transition from \prb space
$\Omega=\{(\gamma,s)\}\equiv \Gamma_k\times \Sigma$ of size $\binom{n}{k}\times 2^n$ to a partition of this space into sets $G(t,x)$ with $0\leq x\leq n,$ $0\leq t\leq min(x,k)$, using the \flw construction. \par
%%%%%%%%%%%%% (\gamma,s)
We denote rv $N(s)$ the number of zeros in signal $s$. To describe its distribution, let us introduce rvs
$N_1(\gamma,s)$ and $N_2(\gamma,s)$: $N_1$ is the number of minuses in locked boxes, i.e. the number of \emph{false minuses}, or equivalently, the number of locks in boxes with minuses; $N_2$ is the number of minuses in unlocked boxes, i.e. the number of \emph{correct minuses}. Then the total number of minuses after testing $N=N_1+N_2$. The rv $N_1$ is a binomial rv with $k$ trials and probability of success $1-a$, the rv $N_2$ is a binomial rv with $n-k$ trials and probability of success $b$. These two random variables are independent, and unless $b=1-a$, rv $N=N_1+N_2$, taking values $0,1,...,n$, is not a binomial rv. Sometimes the distribution of $N$ is called the Poisson binomial distribution. The signal $s=(s_1,...,s_n)$ and the value $N=x$ are \emph{observable} to AT in contrast to the values of $N_1$ and $N_2$, which are not. To stress this distinction, sometimes we use the notation $N(s)$ and $N_1(\gamma,s)$, because AT knows only $s$ and $N$, and only a potential ``observer" who knows \emph{both} $\gamma$ and $s$ knows the values of $N_1$, and $N_2$. This also explains why the \dst $P$ on the set of outcomes $\Omega=\{(\gamma,s)\}\equiv \Gamma_k\times \Sigma$ generally depends on prior $\pi$, but ``marginal" \dsts \ of rvs $N_1,N_2$ and $N$, depend only on parameters $n,k,a,b$. Of course the source of this contrast is the symmetry of testing despite the potential asymmetry of prior $\pi$. \par
%%%%%%% , though even if a DF uses her optimal uniform \dst DF $\pi_*$, values of $N_1$ and $N_2$ are not observable.
We denote by $p_i(j)$ the pmf (probability mass function) of the binomial rvs $N_i, i=1,2$ and $p(j|r,p), j=0,1,...,r$, the pmf of a binomial distribution with $r$ trials and probability of success $p$.
		%the sum of two binomial random variables with different probabilities of success,
 Thus $p_1(j)=p(j|k,1-a)$ and $p_2(j)=p(j|n-k,b)$.  Then the pmf of rv $N$ in problem $A(n,k)$, $g_{A}(x)\equiv g_{n,k}(x), 0\leq x \leq n$, can be calculated by standard \emph{discrete convolution formula}, given below.
%n  problem B the number of locks is rv $K$ with a binomial distribution with $n$ trials and probability of success can be calculated by the second formula below 		%%%%%%%%% fla 1 p 8
\begin{eqnarray}
			P(N=x)\equiv g_{n,k}(x)=\sum_{j}p_1(j)p_2(x-j)\equiv\sum_{t}p_1(x-t)p_2(t). \label{con}
\end{eqnarray}
%%% eto 2 (\ref{P}) %%The summation over $j$ in the convolution formula above is taken over values $j$ such that $0\leq j\leq k, \ 0\leq x-j \leq n-k$. Similar holds for the summation over $t$, where $0\leq x-t \leq k, \ 0\leq t\leq n-k$. Further may omit
In convolution and similar formulas we omit the exact range of summation assuming that all probabilities involved in the sums are well defined. %%%%%%%% f-la moved to Apnd ADst of L \par	
The \flw \dst will also play an important role in our analysis:  the joint \dst of $N_1$ and $N$
%%  \vspace{-.2cm}
\begin{eqnarray}
s(t,x)\equiv P(N_1=t,N=x)=P(N_1=t,N_2=x-t)=p_1(t)p_2(x-t).   \label{stx}
\end{eqnarray}
%%%%%%%%%
%%% !!!! in mtx form  T={s(t,x)=Diagp_1*Shif p-2}+shift p_2 (k+1)*(n+1)=(k+1)*(k+1)*(h+1)*(n+1)
%%% %%% Moved to AFTER to 1600 %% old (10) mvd to APN $0\leq t\leq x\leq n,$ $0\leq t\leq k$.
\vspace{.1cm}
Let us denote
$G(t,x)=\{(\gamma,s): N_1(\gamma,s)=t, N(s)=x\}$, $0\leq x\leq n,$ $0\leq t\leq min(x,k)$. We have the \flw \par
%%%%%%%%%% % \begin{proposition}[My proposition]    Body of the proposition.  \end{proposition}
%%%%    \begin{proof}  \end{proof}
\vspace{.1cm}
\begin{proposition}[]\label{Prop5} %% Proposition 5.
Sets $G(t,x)$ form a partition of set $\Omega=\Gamma\times \Sigma$, and for prior $\pi_{}$,
\begin{eqnarray}
P((\gamma,s)\in G(t,x))=s(t,x), \ \ \ \ \ \ \sum_{0\leq t\leq x}s(t,x)= P(N=x)\equiv g_{n,k}(x).   \label{gtx}
\end{eqnarray}
\end{proposition}
\vspace{.1cm}
%%%%%%%%%%%%%%%% EDIT  and given prior $\pi_*$}}, 
\begin{proof} We will refer to the $M\times N$ stochastic matrix $P=\{p(s|\gamma)\}$ as Basic Table $P$, since it gives a visual representation of \prb space $\Omega$. Given any prior $\pi$, we obtain a \prb \dst on this space $p(\gamma,s)=\pi(\gamma)p(s|\gamma)$. By definition of set $G(t,x)$ and formula (\ref{stx}) we obtain the statement about partition and the equalities in (\ref{gtx}) for any given prior.
 \end{proof}
%%%%%Given uniform prior $\pi_*$ each outcome $(\gamma,s)$ has \prb $p(s|\gamma)/\binom{n}{k}$. Then
\vspace{-.1cm}
%%%%%%%%%%%%%%%%%%%%%%
Proposition \ref{Prop1} for G-LBT model is transformed into the \flw lemma for S-LBT model. We remind that in symmetrical case an optimal \stg of DF, $\pi_*$ is a uniform \dst on set $\Gamma$, i.e. $\pi_*(\gamma)=1/\binom{n}{k}$. \par
%%\vspace{.2cm}
\begin{lemma}\label{Lem1}
 a) For all $\gamma \in \Gamma$ and all signals $s=(s_1,...,s_n)$, such that $(\gamma,s)\in G(t,x)$, the \prb $P(S=s|\gamma)$ is the same and given by formula
%%%%%%%   \ptx            %Using this equality, and the \flw equality,  P(s|N_1=t,N=x)%%%%%%%%%%%%%%
%\vspace{-.3cm}
\begin{eqnarray}
P(S=s|\gamma)=a^{k-t}(1-a)^{t}b^{x-t}(1-b)^{n-k-(x-t)}\equiv p(t,x). \label{ptx}
\end{eqnarray}
%%%%%%%%%%%%%%
b) The values of  $p(t,x)$, $s(t,x)$ and $G(t,x)$ are related as follows
%%%%%%%  %%%%%%%
%\vspace{-.3cm}
\begin{eqnarray}
\binom{n}{k}s(t,x)=p(t,x)|G(t,x)|,\ \text{where}\  |G(t,x)|=\binom{n}{k}\binom{k}{t}\binom{n-k}{x-t}\equiv\binom{n}{x}\binom{x}{t}\binom{n-x}{k-t}.
\label{btx}
\end{eqnarray}
%%%%%
c) Given prior $\pi_*$,  for all $\gamma \in \Gamma$ and all signals $s=(s_1,...,s_n)$, such that
$(\gamma,s)\in G(t,x)$, the posterior \prb $\theta(\gamma|s,\pi_*)$ is the same, i.e. $\theta(\gamma|s,\pi_*)=\theta(t,x)$  and given by formula
%%%%%%%  %%%%%%% %\vspace{-.3cm}
\begin{eqnarray}
\theta(t,x)=p(t,x)/m(x),\ \text{where}\ m(x)=\sum_{(\sigma,s): N(s)=x}p(\sigma|s,\pi_*)=
\sum_{i}p(i,x)\binom{x}{i}\binom{n-x}{k-i}.
\label{ttx}
\end{eqnarray}
%\vspace{-.2cm}
%%%%%%%%%%\varphi $(p(s))\equiv (p(s|\pi_*)) %%%%%%%%%%%    ubral          \equiv p(x)    {0\leq i\leq x%%%%%%
d) Given prior $\pi_*$, for any signal $s=(s_1,...,s_n)$ and any $x=0,1,...,n,$
\begin{equation}
P(T_{i}=0|S=s,N=x)=P(T_{i}=0|S_i=s_{i},N=x)\equiv \alpha_i(s_i,x). \label{alx} %%%A1
\end{equation}
 \end{lemma}
 %%%%%%%%%%%%
%\textbf{Lemma 1, S-LBT} model. \label{Lem1}
%%%%%%%%%% f) \theta(\gamma|s,N_1=t, N=x)=  mvd to AFTER     %%%%%%%% important
\par
Before proving Lemma \ref{Lem1}, we note that in points a) and b) any prior $\pi$ or prior $\pi_*$ is not mentioned at all, but in points c), and d) uniform prior $\pi_*$ plays a crucial role, so, in fact, $P$ in those points is $P_{\pi_*}$ but we keep shorter notation $P$. Thus in points a) and b) the symmetry of testing implies that $P(S=s|\gamma)$ for any $s$ does not depend on a position of $\gamma$ directly, but only through the observable number $N(s)=x$ of zeros and not observable number of false zeros among $k$ locks, i.e. $N_1(\gamma,s)=t$. After the proof we give a simple example to illustrate \clcs \ in formulas in Lemma \ref{Lem1}. %% \par
%%%%%%%%%% This property was also discussed earlier in this Section. NUJNA li last sentesnce ?\vspace{.1cm}
%%%%%%%%%%%%%%% \vspace{.1cm} %Of possible interest in both models are the \emph{aposterior distribution of locks} (ADL)
%%%% $\binom{n}{k}\times 2^n$-dimensional array $B(\gamma|s)$, where $\gamma$ takes all $\binom{n}{k}$ possible values.
\begin{proof} a) The pmf of a Bernoulli rv has a form $p(u)=p^u(1-p)^{1-u}, u=0,1$, and then, using formula (\ref{psg}) in point a) of Proposition \ref{Prop1} we have $\prod_{i\in \gamma}P(s_i|\gamma)=(1-a)^ta^{k-t}$, where $t$ is the value of rv $N_1$. Similarly, $\prod_{i\notin \gamma}P(s_i|\gamma)=b^{x-t}(1-b)^{n-k-(x-t)}$, where $x$ is the value of rv $N$. Combining these two products, we \imd obtain the equality in formula (\ref{ptx}). Thus, we obtained that all entries in Basic Table $P$ contain values $p(i,x)$ for all possible values of $i$ and $x$. On intersection of column $s$ with $N(s)=x$ and row $\gamma$ with $N_1(\gamma,s)=i$ we have value $p(i,x)$. The number of such columns is $\binom{n}{x}$ and the number of such rows is $\binom{x}{i}\binom{n-x}{k-i}$. To obtain the second equality we have $i$ choices for $i$ locks among $x$ minuses and
$k-i$ choices for remaining $k-i$ locks among $n-x$ pluses.\par
\vspace{.1cm}
%%The second equality follows from the equality (\ref{stx}).%%% The first equality in (\ref{sptx})     there is $M=\binom{n}{k}$ of such rows.
b) First, we calculate $|G(t,x)|$, i.e. the number of pairs $(\gamma,s)$ such that $N_1(\gamma,s)=t, N(s)=x$, where $0\leq t\leq k, t\leq x\leq n$. To calculate this number we apply the
combinatorial Product Principle for three stages: first, we have $\binom{n}{k}$ choices to select $k$ positions for $k$ locks. On the second stage we have $\binom{k}{t}$ choices among these positions to select $t$ positions for $t$ zeros in signal $s$, and on the third stage we have $\binom{n-k}{x-t}$ choices to select positions with $x-t$ remaining zeros among $n-k$ positions without locks. The last equality in (\ref{btx}) is a similar application of Product Principle starting with allocation of $x$ zeros among all $n$ positions in signal $s$, then selecting $t$ positions for locks among these $x$ positions, and finally selecting $k-t$ positions for remaining $k-t$ locks among $n-x$ positions for ones in signal $s$. This last equality was in fact proved in a) for $t=i$, where we calculated the number of columns with $N(s)=x$ and the number of rows with $N_1(\gamma,s)=i$. We also have $\sum_{t,x}|G(t,x)|=|\Omega|=MN=\binom{n}{k}2^n$. \par
%%%%%%%%%% let us select an arbitrary prior $\pi$. Then, using formula (\ref{gtx})
To show the first equality in (\ref{btx}) note that by definition of $s(t,x)$ for any arbitrary prior $\pi$ we have
$s(t,x)=\sum_{(\gamma,s)\in G(t,x)}\pi(\gamma)p(s|\gamma)$. By point a) in this sum each $p(s|\gamma)=p(t,x)$. We can select
$\pi=\pi_*$ with $\pi_*(\gamma)=1/\binom{n}{k}$ and then we obtain that $\binom{n}{k}s(t,x)=p(t,x)|G(t,x)|$, i.e.
the first equality in (\ref{btx}). \par
\vspace{.1cm}
%%%%%%%%%%%%Formula (\ref{ptx}) follows from (\ref{psg}) using the definitions of rvs $N_1$ and $N$.\par
c) This point is a reduction of point c) in Proposition \ref{Prop1}.  According to a general formula (\ref{btx}) we have $\theta(\gamma|s,\pi_*)=\frac{\pi_{*}(\gamma)p(s|\gamma)}{p(s|\pi_{*})}$, where
$p(s|\pi_{*}) =\sum_{\sigma}\pi_*(\sigma)p(s|\sigma)$.
We have $\pi_*(\gamma)=\pi_*(\sigma)=1/\binom{n}{k}$ and by point a) for $(\gamma,s)\in G(t,x)$ we have $p(s|\gamma)=p(t,x)$.
Then $\theta(\gamma|s,\pi_*)=p(t,x)/\sum_{\sigma}p(s|\sigma)$, where
in the sum we have a summation only over $(\sigma, s)$ with $N(s)=x$. Again by point a) $p(s|\sigma)=p(i,x)$ for some
$0\leq i\leq min(k,x$). The number of terms with $p(i,x)$ equals $c_i=\binom{x}{i}\binom{n-x}{k-i}$.
Indeed, in our Basic Table $P$, the sum of all these terms over $i$ is a sum over any column under value $x$, since they all have the same sum. The number of terms with $N_1=i, c_i$, given $x$, we can obtain using the
Product Principle for two stages: first, we have $\binom{x}{i}$ choices to select $i$ positions for $i$ locks among
$x$ minuses, i.e. $i$ positions for $i$ false minuses; on the second stage we have $\binom{n-x}{k-i}$ choices to select $k-i$ positions for remaining locks among $n-x$ plus positions. \par
%% $s(t,x)$ by formula (\ref{gtx}).and then summation gives us formula (\ref{ttx}).
Note that in S-LBT we can avoid completely the calculation of $\theta(\gamma|s,\pi)=\theta(t,x)$ because, as in G-LBT, the optimal \stg of AT is based only on $\alpha_i(s_i|\pi)$, which now has a simplified form, given in point d). Thus, formula (\ref{ttx}) is used mainly to check accuracy of other formulas. \par
%%%%%%%%%%%%%%
d)  The equality (\ref{alx}) is intuitively clear. Indeed, the information about $s_i$, at particular site $i$, say $s_i=0$, and how many other sites have negative results is important, but the exact positions of these other sites are not, since all sites identical. The formal proof although is not quite obvious, and it will be given in the Appendix. In Lemma 2 we present the explicit formulas for $\alpha_i(s_i=0,x)\equiv p^{-}(x)$ and $\alpha_i(s_i=1,x)\equiv p^{+}(x)$ based on the equality in (\ref{alx}). 
\end{proof} %%%%%%%%% The points e) and f) of Lemma 1 are proved in the Appendix since they do play a role further.\par
%\vspace{.1cm} %%%% =P(T_i=0|\pi)=1-\beta_i(\pi), i=1,...,n Given a $M$-dimensional row vector of prior \dst $\pi$,
\vspace{-.2cm}
We give a simple example to illustrate \clcs \ in formulas in Lemma \ref{Lem1}. \par
\vspace{.1cm} %%%%%%%% Ex A(3,1) mvd to l 960\newpage
% \begin{proposition}[My proposition]    Body of the proposition.  \end{proposition}
% \begin{example}[My example]    Body of the example.  \end{example}
\begin{example}[\textbf{A(2,1)}]     \label{A21}
Denote two possible positions of a lock: $\gamma_1=(1,0)$ and $\gamma_2=(0,1)$. There are $2^2=4$ possible signals $s$: $s(1)=(+,+)\equiv (1,1), s(2)=(1,0), s(3)=(0,1), s(4)=(0,0)$. Thus $M=2, N=4$. Using formula (\ref{be}) or (\ref{ptx}) we can obtain all $8$ entries of matrix $P=\{p(s|\gamma)\}$. They are given in rows 3, and 4 in Table 1. We have
$p(s(1)|\gamma_1)=P(S_1=1|T_1=1)P(S_2=1|T_2=0)=a(1-b)=e_1.$ All other values of
$p(s(i)|\gamma_1), i=2,3,4$ are: $e_2=ab$, $e_3=(1-a)(1-b), e_4=(1-a)b$.
They are listed in the third row of $5\times 5$ matrix $P$ in Table 1. The values of $p(s(i)|\gamma_2 )$ are listed in the fourth row of matrix in Table 1. They are $e_1, e_3, e_2, e_4$. This Table also reminds us that these are values of $p(t,x)$: $e_1=p(0,0)$, $e_2=p(0,1)$, $e_3=p(1,1)$, $e_4=p(1,2)$.\par
%%%%%%%%%%%%%%%%%%%%%%%%%%%%
The values of $m(x)$ for $x=0,1,2$ are listed in the fifth row of Table 1, and they are obtained using formula in (\ref{ttx}). They are: $m(0)=2e_1$, $m(1)=e_2+e_3$, $m(2)=2e_4$. Then posterior \prbs $\theta(t,x)$ are $\theta(0,0)=1/2$, $\theta(0,1)=e_2/(e_2+e_3)$,  $\theta(1,1)=e_3/(e_2+e_3)$, and $\theta(0,1)=1/(2$.\par
%%%%%%%%%%%% The \dst of $N_1$ is $Bin(1, 1-a)$, the \dst of $N_2$ is $Bin(2, b)$ and then $s(t,x)$  (\ref{stx}).
The \dst $s(t,x)=P(N_1=t, N=x)$, obtained by formula (\ref{stx}), is given in Table 2. Then the \dst of $N$ can be obtained by formula (\ref{con}) or simpler, using formula (\ref{gtx}), i.e. adding values $s(t,x)$ in every column $x$. Thus $g_{2,1}(x)$ with values $(0,1,2)$ has \prbs $(n_0, 2n_1, n_2$.
%%% Let us denote  %$ |G(0,1)|=3, |G(0,2)|=6$, $|G(0,3)|=\varnothing$, and
Using formula (\ref{stx}), we also obtain $G(0,0)=\{(1,1), (2,1)\}$,
$G(0,1)=\{(1,2),(2,3)\}$, $G(1,1)=\{(1,3), (2,2))\}$,
$G(1,2)=\{(1,4),(2,4)\}$. All entries with $t=0$, i.e. from some $G(0,x), x=0,1,2$ are marked in Table 1 with extra $0$. In this example $t=0,1$, so all other entries have $t=1$.
\end{example}
\vspace{.1cm}
%\textbf{Example 2.} $A(7,3), x=5$, Fig. 2.\par  \otimes eto krug s X
\begin{center}
%%%%  ymenshil all na .2cm
%\begin{tabular}{ |m{1.5cm} |m{1.3cm} |m{1.3cm} |m{1.3cm}  |m{1.2cm} |}
%% inFeb\begin{tabular}{ |m{1.7cm} |m{1.5cm} |m{1.5cm} |m{1.5cm}  |m{1.4cm} |}
%% teper use fr T5 line 1213 inus.3
\begin{tabular}
%{ | m{1.2cm} | m{2.0cm}  | m{2.0cm}  | m{2.0cm} |  m{1.2cm}| |m{1.5cm}|}
{|c|c|c|c|c|}    %AS
         \hline
              $N=x $      & $0$      & $1$   &  $1$ & $2$     \\
		\hline
           $\gamma / s(i)= $  & $s(1)=11$ & $s(2)=10$ &$S(3)=01$ & $s(4)=00$    \\
		\hline
         $\gamma_1 =10$ &$e_1,0$  &$e_2,0$  &$e_3$ &$e_4$                \\
		\hline
		$\gamma_2=01$ & $e_1,0$ & $e_3$ & $e_2,0$ & $e_4$                \\
        \hline
		 $m(x)$ & $2e_1$ & $e_2+e_3$     & $e_2+e_3$    & $2e_4$          \\
         \hline
			\end{tabular}
\end{center}
%%%%%%%%%%%%%% .3 to .1
\vspace{.1cm}
\textbf{Table 1.} Matrix $P=\{p(s|\gamma)=p(t,x)$ for Example $A(2,1)$; $e_1+e_2+e_3+e_4=1.$\par
\vspace{.1cm}
\par
The \crsp \ table for the values of $s(t,x)=p_1(t)p_2(x-t)$ has a form  \par
%%\vspace{.1cm} %% .3
\begin{center}
	\begin{tabular}
%{| m{1.5cm} | m{2.3cm}  | m{2.3cm}  | m{2.3cm} |  |m{1.8cm}|}
{|c|c|c|c||c|}
        \hline
      \textbf{$ p_2(1|b)$}  & $b_0=(1-b)$  & $b_1=b$ &    &  \\
		\hline
              $ t/x$      &  $0$         & $1$    & $2$    & \textbf{$p_1(1|1-a)$} \\
		\hline
         $   0   $        & $a_0b_0=e_1 $ & $a_0b_1=e_2$   & $0$  & $a_0=a$  \\
		\hline
         $1    $          & $0$          & $a_1b_0=e_3$    & $a_1b_1=e_4$   & $a_1=1-a$ \\
		\hline
		$g_{2,1}(x)$      & $e_1$   & $e_2+e_3$  & $e_4$    &  $\textbf{1}$\\
         \hline
		\end{tabular}
\end{center}
%\vspace{.2cm}
\textbf{Table 2.} Matrix $T=\{s(t,x)\}=P(N_1=t, N=x)$ for Example $A(2,1).$ \par
\vspace{.3cm}%%%%%%%%%%%% see Sh 10
Let $a=\frac{7}{12}, b=\frac{9}{12}$. Then $(e_1, e_2, e_3, e_4)=(21,5,15,7)/48$,
%% $(n_0,n_1,n_2,n_3)=(7/192,47/576, 31/192, 45/192)$,
$ g_{2,1}(x)=(n_0, 2n_1,n_2)=(21, 20, 7)/48$.\par
%%%%%%%
\vspace{.1cm}
\textbf{Example A(3,1)} is relegated to Subsection 6.1 in Appendix. In contrast to Example A(2,1) sets $G(t,x)$ there have different sizes.
%%%%%%%%%%%%% and $r_B$
\subsection{Key ratio $r(x)$}\label{ratios}
%\subsection{Basic Notation}\label{Bn} use mainly the classical total \prb and Bayes formulas,
%$max(0, x-(n-k))\leq i\leq{min(x,k)},\ \  max(0,x-k \leq t \leq{min(x,n-k))}$.
%%%\\\section{\ \sum %%
\vspace{.1cm}
First, we obtain formulas for $P(S=0)$ and $P(S=1)$, valid for both $A(n,k)$ and $B(n,\lambda)$ models. Given that the \prb of protection for a \emph{particular} box $P(T=1)=t$, where $t=k/n$ for model $A(n,k)$ by formula (\ref{pti}), the test for a box gives plus, $S=1$ or minus, $S=0$ with \prbs
%\vspace{-.2cm}
\begin{eqnarray}
P(S=0)&=&P(T=1)P(S=0|T=1)+P(T=0)P(S=0|T=0)=t(1-a)+(1-t)b,  \notag \\
P(S=1)&=&P(T=1)P(S=1|T=1)+P(T=0)P(1|T=0)=ta+(1-t)(1-b). \label{P}
\end{eqnarray}
%%%%%%%%%%%%%%%%%%%% eto notag with parameters $l=k, 1-a$ and $l=n-k, b$
%$max(0, x-(n-k))\leq i\leq{min(x,k)},\ \  max(0,x-k \leq t \leq{min(x,n-k))}$.
%%%%%%%%%%%%%%%%%%%%%%%%%%%%%%%%%%%%%%%%%%%%%%%%%\\\section{\ \sum
%\vspace{-.2cm}%%%and it is easy to check, using the formula for $P(S)$ in (\ref{P}) and the formulas for $P(C)$ and
%%%%%%%%%\emph{Proof of Proposition 2.} %%\newpage \frac{(n-x)}{x} &=& L2= Prop 2R \vspace{.2cm}
For the model $A(n,k)$  we obtain two different representations for $r_{n,k}(x |a,b)\equiv r_{n,k}(x)\equiv r(x)$ using total \prb formula for two \emph{different} partitions. For the sake of brevity, we denote further
$P(\cdot |N=x)\equiv P(\cdot|x)$.\par
\vspace{.2cm}
 \begin{lemma}\label{Lem2}
  a) The crucial ratio  $r_{n,k}(x)\equiv r(x), \ 0<x<n$, is given by the formula
%% with $n-1=m$, denoted as $r_{m,k,a,b}(x)\equiv r_{m,k}(x)$,\frac{P(C|S'X)}{P(C|SX)}=
\begin{eqnarray}
r(x)\equiv r_{n,k}(x)=\frac{P(T=0|S=0,x)}{P(T=0|S=1,x)}\equiv \frac{p^{-}(x)}{p^{+}(x)}=
\frac{b}{(1-b)}\frac{(n-x)}{x}\frac{g_{n-1,k}(x-1)}{g_{n-1,k}(x)},  \label{rkx}
\end{eqnarray}
\par
%%%%%%%%% \\ \notag $\ds a = b +\frac{c}{d}$.
b) the probabilities used in (\ref{rkx}) $p^{-}(x)\equiv P(T=0|S=0,x)$ and $p^{+}(x)\equiv P(T=0|S=1,x))$ for $0<x<n$ are given by formulas
%$P(S'XA'), P(S'X), P(SXA'), P(SX)$ are given by formulas
%\vspace{-.2cm}
\begin{eqnarray}
p^{-}(x)=\frac{n-k}{x}*b*\frac{g_{n-1,k}(x-1)}{g_{n,k}(x)},\ \ \
p^{+}(x)=\frac{n-k}{n-x}*(1-b)*\frac{g_{n-1,k}(x)}{g_{n,k}(x)};\label{PC1}
\end{eqnarray} \par
c) functions $r_{n,k}(x)\equiv r_{n,k}(x|a,b)>1$ when $a+b>1$, and $r(x)<1$ when $a+b< 1$, and as functions of parameters $a,b$ for all $n,k,0< x <n$ depend only on parameter $ c=\frac{a}{1-a}\frac{b}{1-b}$, and hence satisfy the equality $r_{n,k}(x|a,b)=r_{n,k}(x|b,a)=r_{n,k}(x|\theta,\theta)$, where  $ \theta=\frac{\sqrt c}{1+\sqrt c}$.\
\par
\vspace{.2cm}
%%%%%%%%%%%%% equival TO $\frac{a_0}{1-a_0}=\sqrt c$
d) functions $r_{n,k}(x)$ are monotonically decreasing for all fixed $k, 0<x<n$ when $n$ is increasing; functions $r_{n,n-1}(x)=c$ for all $0< x <n$;  \par
\vspace{.2cm}
e) functions $r_{n,k}(x)$ for $1\leq k<n-1$ are monotonically increasing in $x$ for $0<x<n$, and decreasing in $c$ as a function of $c$.  %%  \par %\vspace{.2cm}
 \end{lemma}
% \begin{proposition}[My proposition]    Body of the proposition.  \end{proposition}
% \begin{example}[My example]    Body of the example.  \end{example}   \begin{proof} \end{proof} \vspace{.2cm} %% .3
%%%%%%%%%% (see details in \cite{lison21}) Presman Math archive (see also
\begin{proof} Points c) and d) are proved in the Appendix. The proof of e), obtained by E. Presman, using the recursive formulas for functions $f_{n-1,k}(x)=\frac{g_{n-1,k}(x-1)}{g_{n-1,k}(x)}$ and induction, is difficult and very technical, and can be found in \cite{pres18}.\par %%%%%%%%% [and an event $D(x)=(N=x)=D)$.]
Now we prove crucial point b), which immediately implies point a), i.e. (\ref{rkx}). \par %%%\vspace{.1cm}
First, let us note that we could define $p^{-}(x)$ (and $p^{+}(x)$) in a more general way as
$p^{-}(x)=P(T_i=0|S_i=0,S_{-i}=s_{-i},x)$ but by point d) of Lemma \ref{Lem1} these \prbs are the same as those given in formula (\ref{rkx}).
Thus $p^{-}(x)\equiv P(T=0|S=0,x)=P(T=0,S=0,x)/P(S=0,x)$, and $p^{+}(x)\equiv P(T=0|S=1,x)=P(T=0,S=1,x)/P(S=1,x)$.
For the numerators we have
\vspace{.1cm}
\begin{eqnarray}
P(T=0,S=0,x)&=&P(T=0)P(S=0|T=0)P(x|T=0,S=0),\notag \\
P(T=0,S=1,x)&=&P(T=0)P(S=1|T=0)P(x|T=0,S=1). \label{PC3}
\end{eqnarray}
%\vspace{-.6cm} %and similarly P2+P4 \vspace{-.3cm}
%% \begin{eqnarray} P(T'SD)=P(T')P(SD|T'), \ \ \ \ P(SD|T')=P(S|T')P(D|T'S). \label{PC4} \end{eqnarray}
We also have $P(T=0)=\frac{n-k}{n}$, $P(S=0|T=0)=b$, and $P(S=1|T=0)=1-b$. Therefore, to prove (\ref{PC1}) it is sufficient to show that
%%the first equality in (\ref{PC1}) and the first equality in (\ref{PC2}).%\vspace{-.4cm}
\begin{eqnarray}
P(x|T=0,S=0)&=&g_{n-1,k}(x-1),\ \ \ \ P(x|T=0,S=1)=g_{n-1,k}(x), \notag \\
P(S=0,x)&=&g_{n,k}(x)*\frac{x}{n}, \ \ \ \ \ \ \ \ \ \ P(S=1,x)=g_{n,k}(x)*\frac{n-x}{n}.  \label{PC2}
\end{eqnarray}
%\vspace{-.2cm} %  PC2 and PC5 combined and\vspace{-.2cm} %\begin{eqnarray}
%P(D(x))|T'S')=g_{n-1,k}(x-1),\ \ \ \ P(D(x))|T'S)=g_{n-1,k}(x).   \label{PC5}
To show the first line of (\ref{PC2}), note that if the total number of minuses is $x$, and a particular box has no lock, $T=0$, and produced minus, $S=0$, then in the remaining $n-1$ boxes there are $k$ locks and  $x-1$ minuses. If the total number of minuses is $x$, and a particular box has no lock, $T=0$, and produced a plus, $S=1$, then in the remaining $n-1$ boxes there are $k$ locks and  $x$ minuses.\par
To show the second line of (\ref{PC2}), note that $P(S=0,x)=P(N=x)P(S=0|N=x)$. We have $P(N=x)=g_{n,k}(x)$ and
$P(S=0|x)=\frac{x}{n}$, the \prb for one minus among $x$ to be in a particular box. Similarly, $P(S=1,x)=P(N=x)P(S=1|N=x)$, and $P(S=1|N=x)=\frac{n-x}{n}$, the \prb for one plus among $n-x$ to be in a particular box. Finally, using the equalities in (\ref{PC3}) and (\ref{PC2}), we obtain  (\ref{PC1}), and therefore formula (\ref{rkx}).
\end{proof} %% \par %%%%%%%%%%%%% \vspace{.1cm}
%% % \begin{proposition}[My proposition]    Body of the proposition.  \end{proposition}
% \begin{example}[My example]    Body of the example.  \end{example} \label{Ex3.6}
\begin{example}[\textbf{A(2,1)}]       %% Lemma 2
Calculation of $r(x), x=1$ for $a=\frac{7}{12}, b=\frac{9}{12}$. By point d) of Lemma \ref{Lem2} for any problem $A(n,n-1)$ we have $r(1)=c=\frac{a}{1-a}\frac{b}{1-b}=\frac{21}{5}=4.2$.
\end{example}
%%  Note that we can not use formula (\ref{rkx}) because we must have $x<n$.     \par
%%%%%%%%%%%%%%%%%% nvx by (15) ???  \label{Ex3.7}
\vspace{.1cm}
\begin{example}[\textbf{A(3,2)}]
Let $a=\frac{7}{12}, b=\frac{9}{12}$. As in previous example by point d) of Lemma \ref{Lem2}, $r(1)=r(2)=c=\frac{7}{5}*\frac{9}{3}=\frac{21}{5}=4.2$.
 \end{example}
%\vspace{.2cm}
\vspace{.2cm}
%Similarly to ref we have (\ref{P}) and (\ref{S}), \frac{k}{n} \vspace{.2cm}
To present the second formula for $r_{n,k}(x)$, for the sake of symmetry, let us introduce the rv $U_1=k-N_1 $,
 the number of locked boxes that were tested plus, i.e., the number of \emph{"correct"\ pluses}, and $U_2=n-k-N_2=n-k-N+N_1$, \crsl the number of \emph{"false"\ pluses}.\par
%\vspace{.1cm} %%\newpage \frac{E(N_2|X/x}{(n-k-E(N_2|X)/(n-x)} \frac{p_{-}(x)}{p_{+}(x)}\equiv
%%% L 3 =  Proposition 3R.
\vspace{.1cm} %%
\begin{lemma}[] \label{Lem3}
  a) The ratio $r_{n,k}(x)$, with $0<x<n$, is given by formula
\vspace{.2cm}
%%%%%%%% \notag \\ p_{-}(x)=\frac{E(N_2|X)}{x}; \ \ \ p_{+}(x)=\frac{E(U_2|X)}{n-x}\equiv \frac{n-k-E(N_2|X)}{n-x},
\begin{eqnarray}
r(x)\equiv r_{n,k}(x)=\frac{P(T=0|S=0,x)}{P(T=0|S=1,x)}\equiv \frac{p^{-}(x)}{p^{+}(x)}=
\frac{n-x}{x}\frac{E(N_2|x)}{E(U_2|x)},
\label{rkx2}
\end{eqnarray}
b) the terms $p^{-}(x)\equiv P(T=0|S=0,x)$ and $p^{+}(x)\equiv P(T=0|S=1,x)$ used in (\ref{rkx2}) and in (\ref{PC1}), are given by formulas
%%%%%%%%%%%%%
\vspace{.2cm}
\begin{eqnarray}
p^{-}(x)=\frac{E(N_2|x)}{x}\equiv \frac{x-E(N_1|x)}{x}; \ \ \ p^{+}(x)=\frac{E(U_2|x)}{n-x}\equiv \frac{n-x-k+E(N_1|x)}{n-x}, \label{px}
\end{eqnarray}
where $E(N_1|x)=\sum_{t}ts(t|x)$, $s(t|x)$ is the conditional \dst
$s(t|x)=P(N_1=t|N=x)\equiv s(t,x)/g_{n,k}(x)$, and the \dst \ $s(t,x)=P(N_1=t, N=x) $ can be obtained by formula (\ref{stx}). \par
 %%%%E(N_2|x)=x-E(N_1|x),\ %%% this is for class Ao\pril 11
%\end{document} Is this the last equality is the same as in (\ref{rkx})? \par \vspace{.2cm}
\end{lemma}
 Proof of Lemma \ref{Lem3} is given in the Appendix.\par
%\newpage %%%%%%%%%%% Recursive formulas for Th 1.\par %%%%%%%%%%% Prop 2A (for $B(n,\lambda)$ moved after RND.
\vspace{.1cm}
% \begin{proposition}[My proposition]    Body of the proposition.  \end{proposition}
% \begin{example}[My example]    Body of the example.  \end{example}
\begin{remark}\label{Rem3} The equalities in (\ref{px}) show that, in contrast to formula (\ref{rkx}), the value of $r_{n,k}(x)$ in (\ref{rkx2}) has a clear interpretation as a \emph{ratio of two proportions} - the expected \emph{proportion of correct minuses}, to the expected \emph{proportion of false pluses}.
Nevertheless, the monotonicity of $r_{n,k}(x)$ in $x$ in point e) of Lemma \ref{Lem2} remains a bit mysterious. Why, if $x$ is increasing, then the efficiency of placing a bomb into a minus box is increasing, even if $x$ exceeds the number $n-k$ of unlocked boxes? In other words, though $p^{-}(x)$ is decreasing with $x$ increasing, $p^{+}(x)$ is decreasing even faster.
%% Computer simulation can help in this regard.\par of the ratio of proportions of "correct"\ minuses to "false"\ pluses
Our conjecture is that this monotonicity is a special case of a more general statistical property related to information obtained by symmetrical testing of identical boxes!
\end{remark}
\par
%%%%%%%%%%% (\ref{pc1}) D(t|x)\equiv In calculation for with $P(C|S'X)$, by symmetry of their \stg R, observing value
%%%%%% given $N=x$ we consider a partition with respect to values of $U_1=j$.... has also an interpretation through
\vspace{.1cm}
%%%%%% %\vspace{.2cm} %textbf{Remark 3.} %\vspace{.2cm} A heuristic explanation \vspace{.2cm}
% \begin{proposition}[My proposition]    Body of the proposition.  \end{proposition}
% \begin{example}[My example]    Body of the example.  \end{example}
\begin{example}[\textbf{A(3,1)}]
First, for $a=\frac{7}{12}, b=\frac{9}{12}$, we calculate $r(x)$, using formula (\ref{rkx}). %%  Zachem oni ?   Values for $g_{3,1}(x)$ were obtained after Table 2 and they are $(7, 47, 93, 45)/192$.
Using the equality (\ref{con}) we obtain that $g_{2,1}(x)=(7, 26, 15)/48.$ Then by formula (\ref{rkx}) we have $r(1)=\frac{9}{3}*\frac{2}{1}*\frac{7}{26}=\frac{21}{13}\approx 1.615$. For $x=2$ we have
$r(2)=\frac{9}{3}*\frac{1}{2}*\frac{26}{15}=\frac{13}{5}=2.6$.\par
%\vspace{.1cm}
Now, we calculate $r(x), x=1,2$, using formulas (\ref{px}) and (\ref{rkx2}). Using Table 2 with values $s(t,x)$,
$t=0,1, x=1,2$, we obtain that
$s(1,1)=E(N_1|x=1)=e_3/3n_1=\frac{5}{47}$ and $s(1,2)=E(N_1|x=2)=2e_5/3n_2=\frac{10}{31}$. Using formula (\ref{px}), we obtain that\par
\vspace{.05cm}
%%%%%%%%%%
$p^{-}(1)=\frac{42}{47}, p^{+}(1)=\frac{26}{47}$, $p^{-}(2)=\frac{26}{31}, p^{+}(2)=\frac{10}{31}$, and $r(1)=\frac{21}{13}\approx 1.615$, $r(2)=\frac{13}{5}=2.6$.
 \end{example}
%%%%%%%%%%%%%%
\vspace{-.1cm}
\section{Theorem 1. Optimal Strategies and Values.}\label{OSV}
By solving the S-LBT model we mean to find an optimal \stg of AT and the corresponding value function.
% i.e., the value of the functional under the optimal \stge.
A \stg of AT $\varphi\equiv \varphi(\cdot |s,m)=(u_1, u_2, \ldots, u_{n}|s,m)$,  $\sum_{i=1}^{n}u_i=m$, is an allocation of $m$ bombs in $n$ boxes, given signal $s$, defined for all $m$ and $s$. Further, we will skip sometimes to indicate dependence on $m$.\par
%%%%%%%%% here a LOT is MISSINg thru s  !!!!!!!!!!!  instead of optimizing
Our first step is to show that, for optimization purposes, the allocation of $m$ bombs for each signal $s$ is the same for all signals with the same value of minuses $N(s)=x$.\par
\vspace{.2cm}
%%%%%%%%%%%%%%%%%%%%%%%%%%%  Formulas (\ref{psg}) and (\ref{psx}) %%%%
% \begin{proposition}[My proposition]    Body of the proposition.  \end{proposition}
% \begin{example}[My example]    Body of the example.  \end{example}
%%%%%%%%%% % \begin{proposition}[My proposition]    Body of the proposition.  \end{proposition}
%%%%    \begin{proof}  \end{proof}
\begin{lemma}[] \label{Lem4} %%% \textbf{Lemma 4=4.1}
Let $B^{-}(s)$ and $B^{+}(s)$ be the sets of minuses and pluses, given signal $s$ with $N(s)=x$.
Let $(u_1, u_2, \ldots, u_{n}|s,m)$ be some allocation of bombs and $U^{-}\equiv U^{-}(\varphi|s)=\{u_i, i\in B^{-}(s)\}$ and $U^{+}\equiv U^{+}(\varphi|s)=\{u_i\in B^{+}(s)\}$ be two sets of the values of $u_i$ in minus and plus boxes. Then any permutations of set $U^{-}, (U^{+})$ among minus (plus) boxes have the same value of destruction.
\end{lemma}
% \par %%\par In other words, It is easy to check that P 5 is equivalent to the claim, that if $\varphi$ is optimal, and $u^+=0$, then $1\leq u^-\leq d$, and if $u^+\geq 1$, then $u^{-}-u^+=d$ or $u^{-}-u^+=d-1$.\par \vspace{.1cm}
\begin{proof} Using the equality in formula (\ref{alx}) and the definitions of $p^{-}(x),p^{+}(x),r(x)$ given in formula (\ref{rxr}) or in (\ref{rkx}), we obtain that, given a \stg $\varphi(\cdot |s,m)=(u_1,...,u_n)$ with $m$ bombs, for any signal $s$ with $N(s)=x$, the value of a \stg $\varphi$ can be represented as
%\vspace{-.2cm} %%%%%%%%%%%%% &=&     w^\varphi(s,x)\equiv
\begin{eqnarray}
w^\varphi(s,x,m)&=&\sum_{i=1}^{n}P(T_i=0|s,x)p(u_i)=\sum_{i=1}^{n}P(T_i=0|s_i,x)p(u_i)= \notag \\
&=&p^{+}(x)[r(x)\sum_{i\in B^{-}(s)}p(u_i)+\sum_{i\in B^{+}(s)}p(u_i)].  \label{ws}
\end{eqnarray}
%%where $w^\varphi(s)\equiv w^\varphi(s,m).$  w^\varphi(x)\equiv
Then, obviously, both sums in formula (\ref{ws}) remain the same under any permutation of $u_i$ in \crsp \ boxes. We denote this common value as $w^\varphi(x,m)$. 
 \end{proof}
Thus, Lemma \ref{Lem4} implies that a \stg in the S-LBT model can be understood as an \lcn of bombs in minus and plus boxes without taking into account the particular positions of these boxes. Therefore we may consider a \stg of AT $\varphi(\cdot |s,m)$ for all $s$ as a collection of functions $w^\varphi(x,m)$ for all $x: 0 \leq x \leq n$. %%% \par
%%%%%%%%%% \equiv v(x) A reader can see this property in Table 1 for Example $A(3,1)$.
We denote $v(x,m)=\sup _{\varphi}w^\varphi(x,m)$, the \emph{value function} over all such strategies, given $m$ and $x$, and $v(m)$, the value function over all \stgs averaged over all possible values of $x$, i.e. $v(m)=\sum_xP(N=x)v(x,m)$.\par
%%%%%%%%%%%%%%%%%%
The remaining question after Lemma \ref{Lem4} is: given the value $x$, how many bombs should go to minus (plus) boxes and how should they be allocated there? The second line in formula (\ref{ws}) gives a hint that, given value $x$, the proportion of the number of bombs placed into a minus box to the number of bombs placed into a plus box is defined by ratio $r(x)=p^{-}(x)/p^{+}(x).$
%\par %\vspace{.1cm} %\textbf{Remark 1.} %% \frac{P(C|S')}{P(C|S)}= ... =\frac{P(C|S'X)}{P(C|SX)}
Under ``normal"\ situations, when $a$ and $b$ are such that tests provide ``truthful"\ information, $r(x)>1$ for all $x$, and therefore, if there is only one bomb, it should be placed into a minus box. According to point c) of Lemma \ref{Lem2} ``normal"\  means that $a+b>1$. When the number of bombs $m$ exceeds $x$, then bombs must be
placed into minus boxes until the number of bombs in each of minus boxes reaches value $d(x)$ and after the next available bomb must go to one of plus boxes empty so far. This claim is a part of Theorem \ref{th:1} which is formulated and proved a bit later. The ``advantage level" $d(x)$, first time described in Introduction is defined by formula %%\par %In a sense, the values $N=x, r, r(x)$ play the role of sufficient statistics in the \par %%%%% Using shorthand notation $P(\cdot |N=x)=P(\cdot |x)$, we have %%%%%%%% P(Ti
%\vspace{-.3cm}
\begin{eqnarray}
d_{}(x)=min(i\geq 1: r(x)q^i<1, i=1,2...), \ \ 0<x<n, \label{dx}
\end{eqnarray}
where $q=1-p$ and $r(x)=r_{n,k}(x)$ is defined by formula (\ref{rxr}) and calculated by (\ref{rkx}) or (\ref{rkx2}).
We assume that $0<p<1$ and $a+b>1$ and hence $r(x)>1$ for $0<x<n$ and $d(x)\geq 1$. If $p=1$ we define $d(x)=1$.\par
%%%%%%%%%%%%%% %%%%%%%%%%%
Formally, given $m$ and $x$, $0<x<n$,  $d$-UAP-\stge, described in Subsection \ref{sym}, with $d=d(x)$ defines a unique \lcn of bombs, given by the tuple $T(x,m)=(l^{-}, e^{-}, l^{+}, e^{+})$, where $l^{-}, l^{+}$ the number of ``complete"\ layers of bombs in the minus and plus boxes, and $e^{-}, (e^{+})$ the number of extra bombs in the ``incomplete"\ minus (plus) layer. All these terms depend on $m,x$, $d=d(x)$ and all parameters $n,k,a,b$, but we do not indicate this explicitly.
%Using shorthand notation $l^{-}=i, e^{-}=e,\ \ l^{+}=j, e^{+}=e'$,
We have $m=m^{-}+m^{+},\ m^{-}=l^{-}*x+e^{-},\ m^{+}=l^{+}*(n-x)+e^{+}$, where $0\leq e^{-}<x, 0\leq e^{+}<n-x$ and $e^{-}*e^{+}=0$.
Thus, if $m^{+}=0$, then $m=m^{-}\leq xd$;\ if $e^{+}>0$, then $e^{-}=0$ and $l^{-}-l^{+}=d$; if $e^{+}=0, l^{+}>0$, then either $l^{-}-l^{+}=d-1, e^{-}\geq 0$ or $l^{-}-l^{+}=d, e^{+}=0$.\par
%%%%%%%%%%
We remind that $p$ is the \prb of explosion of a single bomb in unlocked box, and $p(u)$ is the \prb of explosion of $u$ bombs in unlocked box.\par %% % e*e'=0, \label{m} \end{eqnarray} %%%%%%%%%%%%%%%%%%%%%%%%%%% , $B(n,k)$, F1.
%\vspace{.3cm}  %% $s_1=min(i\geq 1: \Delta (i)r(1)<1)$.
\vspace{.1cm}
 \begin{theorem}[]\label{th:1}
Let, given signal $s$, the total number of minuses $N=x, 0\leq x\leq n$. Then \par
a) if $x=0$ or $n$, then the optimal \stg is to distribute all bombs between boxes UAP and the value function $v(0,m)\equiv v(n,m)$ for a tuple $T(0,m)\equiv T(n,m)$ with $m=n*l+e, \ l=0,1,...,\ 0\leq e<n$, where $l=l^+$ or $l^-$ and $e=e^{+}$ or $e^{-}$, is given by formula
%%%%%%%%%%%%%%%%%%%%%%% (\ref{dx}),
%\vspace{-.2cm}
\begin{eqnarray}
v(0,m)&=&v(n,m)=\frac{n-k}{n}[ep(l+1)+(n-e)p(l)]; \ \  \label{vm0}
\end{eqnarray}
%\vspace{-.2cm} $m=m^{-}+m^{+}, m^{-}=l^{-}*x+e^{-}, m^{+}=l^{+}*(n-x)+e^{+}$, where the
b) if $0<x<n$, then the optimal \stg is to distribute all bombs between minus and plus boxes $d(x)$-UAP,
where $d(x)$ is defined by formula (\ref{dx}).
% $l^{-}=i, e^{-}=e,\ \ l^{+}=j, e^{+}=e'$,   l^{-}, e^{-},l^{+},e^{+}   =min(i\geq 1: r(x)q^i<1),  \label{d*}
%\end{eqnarray} \vspace{-.2cm} %\vspace{-.2cm} \hspace{-3.6cm} =          uniquely)
\emph{The value function $v(x,m)$ for a tuple $T(x,m)=(l^{-}, e^{-},l^{+},e^{+})$ defined by values $m, x$ and $d(x)$-UAP \stge, is given by formula}
%\vspace{-.2cm}
\begin{eqnarray}
v(x,m)=p^{-}(x)[(x-e^{-})p(l^{-})+e^{-}p(l^{-}+1)]+p^{+}(x)[e^{+}(n-x-e^{+})p(l^{+})+e^{+}p(l^{+}+1)].\label{vmx}
\end{eqnarray}
%%%%%%%%%%%%%%%
c) \emph{The value function $v(m), m=1,2,...$ is given by formula}
%\vspace{-.2cm}
\begin{eqnarray}
v(m)=\sum_{x=0}^{n}P(N=x)v(x,m). \label{vm}
\end{eqnarray}
 \end{theorem}
%%%%%%%%
\begin{proof}[Proof of Theorem \ref{th:1}]
%\vspace{.1cm}
We will prove the optimality of \stgs described in Theorem \ref{th:1}, showing that any other \stg does not satisfy a natural and obviously necessary condition: \emph{with an optimal allocation it is impossible to increase the payoff by moving a bomb from one box to another.} This condition is used in Lemmas \ref{Lem5} and \ref{Lem6}. \par
%\vspace{.2cm}
%%%%%%% \subsection{Propositions 4 and 5}\label{Pr45} %%%%%%%%%%%%%%%%
Let $J$ be a subset of boxes and $D(J)$ be an event that \emph{all} boxes in $J$ are destroyed, $D_i$ an event that box $i$ is destroyed. Then, given \stg $\varphi$, we have $ w^\varphi(x,m)=\sum_{i=1}^n P(D_i|u_i,s_i,x)$.
%%%%%%%%%%% \ \  \ \ v_k(m)=\sup_{\varphi}w_k^\varphi, $D(x)\equiv D$.
%%%%%%%%% removed kusok about Additiv Property \par %%%%%%%%%%% where $\alpha (i)=(1-(1-p)^i)/p, i=1,2,...,$.  (16)
%%% and total \prb formula imply the \flw formula for the  conditional \prb of the destruction of a particular box with $u$ bombs, and for any event $F$ generated by testg (signals), $P(C|u,F)=P(C|u,T=0)P(T=0|F)=p(u)P(T=0|F), \ u\geq 1.$ Using this formula
The conditional \ndpc of testing and explosions, formula (\ref{be}), and the definitions of $r(x)$, $p^{-}(x)$  and $p^+(x)$, imply the equalities:
%% 1-(1-p)^i=p\frac{1-(1-p)^i}{p}=p\alpha(i),  \notag \\ &=&  %\vspace{-.2cm}
\begin{eqnarray}
P(D|u,S=1,x)&=&P(T=0|S=1,x)P(D|u,T=0)=p^+(x)p(u),\ \  \notag \\
P(D|u,S=0,x)&=&P(T=0|S=0,x)P(D|u,0)=p^{-}(x)p(u)=r(x)p^+(x)p(u), \label{pcalp}
\end{eqnarray}
where $p(u)=1-q^u$. %%%% u,T=0 rediced to ,0
%%%%%%%%%%%%% ERik $D$ справедливо равенство $$ P(C(i)D) = p\alpha_i (p) P(T^prime D).$$ (\ref{bgt})
%% NET ETO NE NADO %% P(C(i)|T')&=&1-(1-p)^i=p\frac{1-(1-p)^i}{p}=p\alpha(i), %\vspace{.2cm} %\newpage (separately)
%\end{proof}
The next Lemma justifies our claim that the optimal \stgs in all problems are UAP in minus boxes, and as well in plus boxes.
 \par
 \vspace{.1cm} %%% % \begin{proposition}[My proposition]    Body of the proposition.  \end{proposition}
% \begin{example}[My example]    Body of the example.  \end{example}
\begin{lemma} [] \label{Lem5}
Let $\varphi(x)=(u_i,i=1,2,...,n)$ be an optimal \stg for all signals $s$ with $N(s)=x$. Then $|u_r-u_t|\leq 1$ when the signals in boxes $r, t $ have the same sign.
\end{lemma}
\par
%%%%%%%%%\par for both functionals \vspace{.2cm}
In other words, Lemma \ref{Lem5} states that if the number of bombs in two minus boxes (or in two plus boxes) is different by more than one, then this \stg can not be optimal, and therefore the optimal \stg is UAP. This statement holds for all $A$ and $B$ problems, but we give the proof in the Appendix only for model $A$. %%%%%% We prove only a) for the case $F_1$. The proof for $F_2$ is carried in \cite{soli18}. \par \vspace{.2cm} Proof of Lemma \ref{Lem5} is given in the Appendix.\par
%\newpage %%%%%%%%%%% \vspace{.1cm}
%%ls have the \flw property:  \emph{each extra bomb should bring the maximum possible increment to the functional }and hence each a choice between placing an extra bomb into minus or plus box.\par \alpha(\par such \stg \ "$(d,x)$-uniform"\$F_1$.
Our next Lemma almost implies Theorem \ref{th:1}.\par
\vspace{.1cm} %%ls have
% \begin{proposition}[My proposition]    Body of the proposition.  \end{proposition}
% \begin{example}[My example]    Body of the example.  \end{example} \ref{Lem5}
 \begin{lemma}[]\label{Lem6}  %%\textbf{Lemma 6.}
Let $\varphi(x)=(u_l,l=1,2,...,n)$ be a \stge, $0<x<n$,  $u^-=i$ be the number of bombs in some minus box, $u^+=j$ be the number of bombs in some plus box, and $d(x)=d$ is defined by formula (\ref{dx}).
%%%%%%%%%%% $p^+(x)=P(T'|+), p^-(x)=P(T'|-)=r(x)p^+(x)$, %\vspace{-.3cm} \begin{eqnarray} d_{}(x)=min(i\geq 1: r(x)q^i<1).  \label{d} %\end{eqnarray}
Then, if $i-j>d$ or, if $j\geq 1$ and $i-j<d-1$, then \stg  $\varphi$ is not optimal, or, equivalently,
if $\varphi$ is optimal, and $j=0$, then $1\leq i\leq d$, and if $j\geq 1$, then $i-j=d$ or $d-1$.
\end{lemma}
 \par
%%%%%%%%%\par In other words, It is easy to check that P 5 is equivalent to the claim, that if $\varphi$ is optimal, and $u^+=0$, then $1\leq u^-\leq d$, and if $u^+\geq 1$, then $u^{-}-u^+=d$ or $u^{-}-u^+=d-1$.\par
\vspace{.1cm}
Proof of Lemma \ref{Lem6} is given in the Appendix.\par
%%%%%%%%IC. [Interchange Condition.] \subsection{Proof of Theorem 1. $A(n,k,m)$}\label{Sol}
\vspace{.1cm}
Now we can finish the proof of Theorem \ref{th:1}.\par %%% $l^{-}=i, e^{-}=e,\ \ l^{+}=j, e^{+}=e'$,
Point a). If $x=0$ or $n$, then all boxes have the same sign and then Lemma \ref{Lem5} implies that it is optimal to distribute all bombs between all boxes UAP. When $m=n*l^{}+e^{}$, where $0\leq e^{}<n$, then UAP means that $e^{}$ boxes have $l^{}+1$ bombs each, and $n-e^{}$ boxes have $l^{}$ bombs each. By formula (\ref{pti}) the \prb that a particular box has no lock is $\frac{n-k}{n}$. Then, using the last equality in formula (\ref{be}), we obtain that the expected damage in all $n$ boxes is $\frac{n-k}{n}[e^{}P(D|l^{}+1,T=0)+(n-e^{})P(D|l^{},T=0)]=\frac{n-k}{n}[e^{}p(l^{}+1)+(n-e^{})p(l^{})])$, i.e., $v(0,m)=v(n,m)$ is given by formula (\ref{vm0}).\par
\vspace{.1cm}
%%%%%%%%%%%% in both problems           allocation e^{-} $l^{-}=i, e^{-}=e,\ \ l^{+}=j, e^{+}=e'$,
Point b). Let $0< x<n$. To simplify presentation we assume that $r(x)q^{d(x)-1}>1$ and $r(x)q^{d(x)}<1$. Then Lemma \ref{Lem6} implies that there is a unique optimal $d(x)$-UAP \stg according to which all boxes to be filled as follows: first all minus boxes are filled until the number of bombs in each of minus boxes reaches value $d(x)$ or until all bombs are exhausted. After that all available bombs go to plus boxes until each get one bomb. After that all available bomb go to minus boxes now filled up to level $d(x)+1$, then extra bombs are switched to plus boxes, and so on. Thus any values of $x,m$, and $d(x)$ uniquely define tuple
$T(x,m)=(l^{-}, e^{-},l^{+},e^{+})$ with $m^{-}=l^{-}*x+e^{-},\ m^{+}=l^{+}*(n-x)+e^{+}$, where $0\leq e^{-}<x, 0\leq e^{+}<n-x$ and $e^{-}*e^{+}=0$.
%The optimality of $d(x)$-\stg is proven. We mentioned after the proof of Lemma 6 that
If $r(x)q^{d(x)-1}=1$, then there are other but very similar optimal strategies. \par %%% Of course, they have the same value function $v(x,m)$. %%%%%We also proved in Lemma 6 that $e*e'=0$. Then, the each of $e$  minus boxes have $i+1$ bombs each, and $x-e$ minus boxes have $i$ bombs each, and in plus boxes $e'$  boxes have $j+1$ bombs each, and $n-x-e'$ boxes have $j$ each.
Now, using the definitions of $r(x)$, $p^{-}(x)$ and $p^+(x)$, the expected damage in all $x$ minus boxes is
%%%%%%%%%%%% in all $n$ boxes is\\
\vspace{.1cm}
\begin{eqnarray}
(x-e^{-})P(D|l^{-},S=0,x)+e^{-}P(D|l^{-}+1,S=0,x)=p^{-}(x)[(x-e^{-})p(l^{-})+e^{-}p(l^{-}+1)]).\notag
%\label{vmin}
\end{eqnarray}
Similarly, he expected damage in all $n-x$ plus boxes is
\vspace{.1cm}
%%%%%%%%%%%% in all $n$ boxes is\\ and using the equality $p^{-}(x)=p^{+}(x)r(x)$,
\begin{eqnarray}
(n-x-e^{+})P(D|l^{+},S=1,x)+e^{+}P(D|l^{+}+1,S=1,x)=p^{+}(x)[(n-x-e^{+})p(l^{+})+e^{+}p(l^{+}+1)].\notag
%%\label{vplus}
\end{eqnarray}
Adding the two formulas above,  we obtain that the expected damage in all $n$ boxes, i.e. $v(x,m)$, is given by formula (\ref{vmx}). Point b) is proved. Point c) is straightforward.\par
\end{proof}
%%%%%%%%%%%%%%%%% ??? If $e\neq 0$ then similar ...give formula in (\ref{vmx}).\par
%\vspace{.3cm} \textbf{Remark 4.}\label{Rem4}\emph{ Note that the threshold value $d(x)$ depends on $n$ and $x$ only through the value of $r(x).$}\par
\vspace{.1cm}
% \begin{proposition}[My proposition]    Body of the proposition.  \end{proposition}
% \begin{example}[My example]    Body of the example.  \end{example}
\begin{remark}\label{Rem4}
For computational purposes the formulas in (\ref{vm0}), (\ref{vmx}) and (\ref{vm}) can be represented recursively in $m$.
\end{remark}%\par \vspace{.2cm}
%%%%%%%%%%%%%%%%More Examples}\label{moex}
%%\end{document}
\subsection{Application of Theorem \ref{th:1}}\label{moex}
%%%%%%%%%%%%% | 6, & 4
Theorem \ref{th:1} opens the path to describe the explicit solutions for any  $A(n,k,m)$ model in the following, easily programmable steps:\par
\vspace{.1cm}
1. Calculate \dsts  $g_{n,k}(x)$ and $g_{n-1,k}(x)$ using formula (\ref{con}) based on two binomial \dsts $p_1(i)$, $Bin(k,1-a)$ and $p_2(i)$, $Bin(n-k,b)$.\par
\vspace{.1cm}
%%%%%%% functions
2. Calculate values $p^{-}(x), p^+(x)$ and then $r_{n,k}(x)$ for each $x=1,2,...,n-1$ using formulas (\ref{PC1}) and (\ref{rkx}), or (\ref{px}) and (\ref{rkx2}). It is also possible to calculate $r_{n,k}(x)$ using formula (\ref{li2}), proved in Subsection 6.2 in Appendix. \par %This latter formula may have computational advantage since it uses a combined characteristic the quality of testing $c=\frac{a}{1-a}\frac{b}{1-b}$ instead of two separate parameters $a$ and $b$ used in
\vspace{.1cm} %%%%%%%%
These two stages do not depend on quality of bombs, i.e. parameter $p$, the \prb of explosion, or function $p(u)$, and thus depend only on parameters $n,k,a,b$. Parameter $p$ is involved only in steps step 3, 4 and 5.\par
\vspace{.1cm}
3. Calculate advantage values $d(x)$ for each $x=1,2,...,n-1$ using formula (\ref{dx}).\par
\vspace{.1cm}
4. Obtain optimal $d(x)$-UAP \stg of AT, $\varphi_{opt}(x,m)$, i.e. obtain tuple
$T(x,m)=(l^{-}, e^{-},l^{+},e^{+})$ for each $x=0,1,...,n$ and $m\geq 1$, using values $d(x)$. \par
\vspace{.1cm}
5. Calculate game values $v(x,m)$ and $v(m)$ by formulas given in Theorem \ref{th:1}. \par
\vspace{.1cm}
We illustrate these steps for some examples.\par
\vspace{.2cm} %% \ref{Ex3.6}
\textbf{Calculations for Example $A(2,1)$.}\par
\vspace{.1cm}
In this example we use values $a=\frac{7}{12}, b=\frac{9}{12}$ and $p=.6$. We have $r(1)=c=4.2$ by point d) of Lemma 2 valid for any problem $A(n,n-1)$. Using formula (\ref{dx}), $d_{}(x)=min(i\geq 1: r(x)q^i<1, i=1,2...)$ we obtain that
%need to calculate  values $g(i|x)=r(x)q^i$ for integers $i=1,2,...$, $q=1-p$ and $x=1$. We have $g(1|1)=\frac{42}{25},$ $g(2|1)=\frac{84}{125}$ and therefore.. for a general $A(n,k)$ model
$d(1)=2$. For each  $x$ and $m$ according to Theorem \ref{th:1} we can obtain a tuple
$T(x,m)=(l^{-}, e^{-},l^{+},e^{+})$ describing the allocation of bombs between boxes using $d(x)$-UAP \stge. For $x=0,2$ this \stg is UAP, so all boxes are filled sequentially. For $x=1$ we can just give two values $(i,j)$ - the number of bombs in a minus and plus boxes. Since $d(1)=2$, for $m=1,2,3,4,5$ these pairs are: $(1,0)$, $(2,0)$, $(2,1)$, $(3,1)$ and $(3,2)$.
Using formulas (\ref{vm0}),(\ref{vmx}),(\ref{vm}) and  $g_{2,1}(x) =(21,78,45)/144)$, we obtain values given in Table 3.  \par
\vspace{.1cm}
%%%%%%%%%%%%\vspace{.2cm} %%%%%%%%%%%%%% BELOW
%%%%%%%%%%%  Table 3 Below Sh 22  %%%%%%%%%%%%%%% p(4)=.974
%%% probuyu \centering isted  \begin{center}
%%\centering
\begin{center}
	\begin{tabular}
%{|m{0.8cm} | m{0.8cm} | m{0.8cm} | m{0.8cm} | m{0.8cm}  | m{0.8cm} | }
{|c|c|c|c|c|c|}                    %AS
        \hline
        $x/m$      &  $ 1 $  & $ 2 $   & $ 3 $   & $ 4 $   & $ 5 $ \\
		\hline
         $0, 2 $   & $.3$   & $.6$    & $.72$   & $.84 $  & $ .888$ \\
		\hline
         $ 1 $     & $.485$ & $.6 $   & $.794 $  & $.84 $  & $.918$ \\
		\hline
         $v(m)$    & $ .4$  & $.6 $   &  $.766 $ &$.84 $   & $.904 $\\
       	 \hline
		\end{tabular}
\end{center}
\vspace{.2cm}
\textbf{Table 3.} Values of $v(x,m)$ and $v(m)$ for Example $A(2,1).$\par
%%%%%%%%%%% g_{2.1}(x)=(21,78,45)/144$
\vspace{.2cm}               %%%\newpage
Example A(3,1) is considered in Appendix.\par
\vspace{.1cm}
% \begin{proposition}[My proposition]    Body of the proposition.  \end{proposition}
%
\begin{example}[ A(7,3)]
In this example with parameters $a=7/12, b=9/12$ and $p=.6$ the program in Matlab gives values
$r_{7,3}(x),x=1,...,6: (1.615, \	1.748,\	1.932,\	2.174,\	 2.403,\	2.6)$ with
$d_{7,3}(x)=1$ for $x=1,...,5$, $d(6)=2$ and the values of $v(x,m)$ for $x=0,1,2,..,7$:
$v(0,m)\equiv v(7,m)=\frac{n-k}{n}p*m=.343m$ for $m=1,2,...,7$, $v(1,1)=pp^-(1)=.509$, $v(1,m)=p[p^{-}(1)+p^{+}(1)(m-1)]$ for $m=2,...,7$, $v(2,m)=pmp^{-}(2)=.509$ for $m=1,2$ and $v(2,m)=p[mp^{-}(2)+(m-2)p^{+}(2)]$ for $m=3,...,7$ and in the same pattern for $x=3,...6$ with $1\leq m \leq7$. The eight values $v(x,5)$ for $x=0,1,...,7$ are
$(1.714, 1.770,  1.964, 2.158, 2.352, 2.545, 2.545, 2.545)$ and $v(5)=2.348$.\par
%%The eight values $v(x,7)$ for $x=0,1,...,7$ are $...$ and $v(7)=...$.\par
%%%%%% $m=5$ and $m=7$ and 1.714 1.770  1.964 2.158 2.352 2.545 2.545 2.545
When $m>7$ we have to take into account the values of $d(x)$. For example for $x=2$ with $d(2)=1$ for $m=15$ we have the tuple
$T(2,15)=(2,1;2,0)$. For $x=6$ with $d(2)=6$ for $m=15$ we have the tuple $T(6,15)=(2,2;1,0)$, in other words two minus boxes have 3 bombs, four have two and one plus box has one bomb. The eight values $v(x,15)$ for $x=0,1,...,7$ are
$(3.415, 3.441, 3.439, 3.436, 3.431, 3.426, 3.480, 3.415)$ and $v(15)=3.437$.
\end{example}
%We can see that there is no much difference between $v(7)=...$ and $v(15)=...$ because marginal utility of extra bomb in a box is decreasing. %\newpage %%\par %% \vspace{.2cm} We will analyze the examples for problems $A(3,1)$ with $x=2$ and $A(7,3)$  with $x=5$, and parameters $a=7/12, b=9/12$ for both examples.
\section{Some Considerations and Generalizations. Open Problems}\label{GEN}
%%%%%%%%%%%%%%%
It is worth mentioning that G-LBT ("Defense-Rebels model") has a more general character that its title suggests. We mentioned a few interpretations in Introduction. Here we mention a few more. American presidential elections with the distribution of resources between battleground states with polls serving as testing tools fit the LBT  model. In repair/maintenance models one may consider locks as hidden defective blocks and bombs as some maintenance/repair units. In biology and medicine boxes can represent parts of the body (organs) and bombs are units of treatment (chemo, radiation).\par
The G-LBT model can be extended in many different directions. The first natural step is to introduce a universal LBT model, where DF and AT have at their disposal some sets of locks and bombs, possibly with different properties of protection and destruction. An even more general model would be to allow DF and AT to have some limited resources that can be divided between production, allocation and testing. We assumed, for simplicity, that the explosions of different bombs are \ndp, but in a more general model this assumption can also be weakened.\par
%%%%%%%%%%
In our G-LBT model the structure of boxes played no role, but it is possible to consider similar games on directed graphs, or to introduce the ``delivery" cost for bombs and locks.\par
%%%%%%%%%%%%%%%%
One of the possible next steps is to introduce a dynamical aspect where an interaction between DF and AT evolves in time.
Such a universal LBT model will have common features and in a sense will be a very broad generalization of the well-known model in Applied Probability -- the Multi Armed Bandit problem, see (Lattimore et all, 2020) \cite{lasz19} and (Presman, Sonin, 1990) \cite{prso90}. The G-LBT model also is related to Search theory, see (Garnaev, 2000) \cite{Gar00} and (Alpern et all, 2013) \cite{Alp13}. \par
%%\newpage %\section{Some Technical Proofs }\label{SOTEP}
The main open problem is to find a general \lgt \ to be implemented in Stage 3 described in Section \ref{STG}. The other problem is to find explicit description of optimal \stgs \  and value functions for specific models. More details can be found in \cite{sw19} and \cite{son22}.
\vspace{-.1cm}
%%%%%%%%%%%%%of Lemmas 2, 3,  and Theorem 1%\vspace{.2cm} \textbf{Example 1.} $B(3,1), x=2$.\par Not Finished yet.
%\vspace{.2cm} \textbf{Example 2.} $B(7,3), x=5$.\par Not Finished yet. vspace{-.2cm}
\section{Appendix}\label{APN}
%\subsection{ The Aposterior Distribution of Locks (ADL).} %%%% (Li Liu). %\vspace{-.2cm}
\subsection{Example A(3,1)}\label{A31}
% \begin{proposition}[My proposition]    Body of the proposition.  \end{proposition}
%
\begin{example}[\textbf{A(3,1)}]
%Using formula (\ref{con}), we can obtain the \dst of $g_{3,1}(x)=(7, 47, 93, 45)/192.$  Then by formula (\ref{con}) we obtain values for $x=0,1,2,3$ values for $p(x)=( ,  , , , )/$%By point b) of Proposition 2 for any problem $B(n,n-1)$ we have $r(1)=c=\frac{21}{5}=4.2$  \par
Denote three possible positions of a lock: $\gamma_1=(1,0,0)$, $\gamma_2=(0,1,0)$, and $\gamma_3=(0,0,1)$. There are $2^3=8$ possible signals $s$: $s(1)=(+,+,+)\equiv (1,1,1), s(2)=(1,1,0), s(3)=(1,0,1), s(4)=(0,1,1)$, $s(5)=(1,0,0), s(6)=(0,1,0), s(7)=(0,0,1)$,  $s(8)=(0,0,0)$. Using formula (\ref{be}) or (\ref{ptx}) we can obtain all 24 entries of matrix $P=\{p(s|\gamma)\}$. They are given in rows 3, 4 and 5 of $6\times 9$ matrix $P$ in Table 4. We have
$p(s(1)|\gamma_1 )=P(S_1=1|T_1=1)P(S_2=1|T_2=0)P(S_3=1|T_3=0)=a(1-b)^2=e_1.$ All other values of
$p(s(i)|\gamma_1 ), i=2,3...,8$ are: $e_2=ab(1-b), e_2$, $e_3=(1-a)(1-b)^2, e_4=ab^2$, $e_5=(1-a)b(1-b),$ $e_6=(1-a)b^2$.
They are listed in the third row in Table 4, and the values of $p(s(i)|\gamma_j )$ for $j=2,3$ are listed in rows 4 and 5. They are permutations of values $e_i$ in the third row. As in the Table 1, these are values of $p(t,x)$:
$e_1=p(0,0)$, $e_2=p(0,1)$, $e_3=p(1,1)$, $e_4=p(0,2)$, $e_5=p(1,2)$, $e_6=p(1,3)$.\par \par
%%%%%%%%%%%%
The \prbs of $p_{}(s,\pi_*)=p(x,\pi_*)\equiv n_x$ for $x=0,1,2,3$ are listed in the sixth row, and they are obtained using prior $\pi_*$, i.e. in this example by averaging $p(s|\gamma)$ with coefficients $1/3$ for each $\gamma$. They are:  $n_0=e_1$, $n_1=(2e_2+e_3)/3$, $n_2=(e_4+2e_5)/3$, $n_3=e_6$.\par
%%%%%%%%%%%% The \dst of $N_1$ is $Bin(1, 1-a)$, the \dst of $N_2$ is $Bin(2, b)$ and then $s(t,x)$  (\ref{stx}).
The \dst $s(t,x)=P(N_1=t, N=x)$, obtained by formula (\ref{stx}), is given in Table 5. Then the \dst of $N$ can be obtained by formula (\ref{con}) or simpler, using formula (\ref{gtx}), i.e. adding values $s(t,x)$ in every column $x$. Thus $g_{3,1}(x)$ with values $(0,1,2,3)$ has \prbs $(n_0, 3n_1, 3n_2, n_3)$. The values for nonzero $p(t,x)$ are: $p(0,0)=e_1$, $p(0,1)=e_2$,
$p(1,1)=e_3$, $p(0,2)=e_3$, $p(1,2)=e_5$, $p(1,3)=e_6$. We also have $n_0+3n_1+3n_2+n_3=e_1+2e_2+e_3+e_4+2e_5+e_6=1$.
%%% Let us denote  %$ |G(0,1)|=3, |G(0,2)|=6$, $|G(0,3)|=\varnothing$, and
Using formula (\ref{stx}), we also obtain $G(0,0)=\{(1,1), (2,1),(3,1)\}$,
$G(0,1)=\{(1,2),(1,3),(2,2),(2,4)$,$(3,3),(3,4)\}$, $G(1,1)=\{(1,4), (2,3),(3,2)\}$,
$G(0,2)=\{(1,5),(2,6),(3,7)\}$, $G(1,2)=\{(1,6),(1,7),(2,5),(2,7)$,$(3,5),(3,6)\}$, $G(1,3)=\{(1,8), (2,8),(3,8)\}$. All entries with $t=0$, i.e. from some $G(0,x), x=0,1,2,3$ are marked in Table 1 with extra $0$. In this example $t=0,1$, so all other entries have $t=1$.
 \end{example} %% \par \vspace{.2cm}
%\textbf{Example 2.} $A(7,3), x=5$, Fig. 2.\par  \otimes eto krug s X
\begin{center}
%     \begin{tabular}{ | m{1.8cm} | m{1.8cm} | m{1.8cm} | m{1.8cm} | m{1.8cm}  | m{1.8cm} | }
	\begin{tabular}{ |m{1.7cm}  |m{.9cm}  |m{.9cm}  |m{.9cm}  |m{.9cm}  |m{.9cm} |m{.9cm} | m{.9cm} | m{.9cm}|}
         \hline
              $N=x $      & $0$      & $1$   &  $1$ & $1$  & $2$  &  $2$  &  $2$  &  $3$   \\
		\hline
           $\gamma / s(i)= $  & $1; 111$ & $2, 110$ &$3, 101$ & $4, 011$  & $5, 100$&$6, 010$&$7, 001$ &$8, 000$   \\
		\hline
         $\gamma_1 =100$ &$e_1,0$  &$e_2,0$  &$e_2,0$ &$e_3$ &$e_4,0$ & $e_5$ &$e_5$ & $e_6$\\
		\hline
		$\gamma_2=010$ & $e_1,0$ & $e_2,0$ & $e_3$ & $e_2,0$ & $e_5$  &  $e_4,0$ & $e_5$ & $e_6$\\
        \hline
		$\gamma_3=001$ & $e_1,0$ & $e_3$ & $e_2,0$ & $e_2,0$ & $e_5$ &$e_5$ & $e_4,0$ & $e_6$\\
         \hline
         $p(s|\pi_*)$ & $n_0$ & $n_1$     & $n_1$    & $n_1$   & $n_2$  &  $n_2$ & $n_2$ & $n_3$\\
         \hline
			\end{tabular}
\end{center}
%%%%%%%%%%%%%%
\vspace{.2cm}
\textbf{Table 4.} Matrix $P=\{p(s|\gamma)$ for Example $A(3,1)$. \par
\vspace{.1cm} \par
The \crsp \ table for the values of $s(t,x)=p_1(t)p_2(x-t)$ has a form  \par
\vspace{.1cm}
\begin{center}
	\begin{tabular}{ | m{1.5cm} | m{2.3cm}  | m{2.3cm}  | m{2.3cm} |  m{1.5cm}| |m{1.8cm}|}
        \hline
      \textbf{$ p_2(1|b)$}  & $b_0=(1-b)^2$  & $b_1=2b(1-b)$ & $b_2=b^2$  &   &  \\
		\hline
                      $ t/x$ &   $0$    & $1$    & $2$    & $3$  & \textbf{$p_1(1|1-a)$} \\
		\hline
         $   0   $        & $a_0b_0=e_1 $   & $a_0b_1=2e_2$  & $a_0b_2=e_4$    & $0$  & $a_0=a$  \\
		\hline
         $1    $         & $0$          & $a_1b_0=e_3$      & $a_1b_1=2e_5$      & $a_1b_2=e_6$  & $a_1=1-a$ \\
		\hline
		$g_{3,1}(x)$  & $e_1=n_0$   & $2e_2+e_3=3n_1$ & $e_4+2e_5=3n_2$  & $e_6=n_3$ & $\textbf{1}$\\
         \hline
		\end{tabular}
\end{center}
%\vspace{.2cm}
\textbf{Table 5.} Matrix $T=\{s(t,x)\}=P(N_1=t, N=x)$ for Example $A(3,1).$ \par
\vspace{.3cm}%%%%%%%%%%%% see Sh 22
Let $a=\frac{7}{12}, b=\frac{9}{12}$. Then $(e_1, e_2, e_3, e_4, e_5,e_6)$=$(7,21,5,63,15,45)/192$, \newline
%% $(n_0,n_1,n_2,n_3)=(7/192,47/576, 31/192, 45/192)$,
$ g_{3,1}(x)=(n_0, 3n_1,3n_2, n_3)=(7, 47, 93, 45)/192$.\par
%%%%%%% Fr Excel sh 22 e_i  7/192    7/64=21/192    5/192   63/192    15/192   45/192   \vspace{-.2cm}
% g(3,1)   7/192	  47/192	  31/64 	  15/64
\subsection{ Some Technical Proofs }\label{SOTEP}
%%%%%%%%%%%%
\emph{Proof of point d) of Lemma \ref{Lem1}.} We will prove formula (\ref{alx}), using the formula (\ref{psx}) from Proposition \ref{Prop4}. %, proved in the proof of Lemma 1.
The left side of formula (\ref{alx}) can be written as %% Feb 2 vstavki
$P(T_i=0)P(S_i=s_i,S_{-i}=s_{-i}, N=x|T_i=0)/P(S=s,N=x)$, where $S_{-i}$ is vector $S$ without coordinate $S_i$, the similar
notation is used for $s_{-i}$. Let $s_i=0$. Using the trivial equality $P(AB|C)=P(A|C)P(B|AC)$, we can rewrite
$P(S_i=s_i,S_{-i}=s_{-i}, N=x|T_i=0)$ as $P(S_i=0, N=x|T_i=0)P(S_{-i}=s_{-i}, N=x|S_i=0,T_i=0)$. By formula (\ref{be}) $P(S_i=0, N=x|T_i=0)=P(S_i=0|T_i=0)=b$. Now let us show that
$P(S_{-i}=s_{-i}, N=x|S_i=0,T_i=0)=g_{n-1,k}(x-1)/\binom{n-1}{x-1}.$ Indeed, we can apply formula (\ref{psx}) to a situation with $n-1$ boxes, $N=x-1$, and $k$ locks. Then the denominator $P(S=s,N=x)=P(N=x)/\binom{n}{x}$.\par
%we can replace $P(s)=P(N=x)/\binom{n)}{x}$.
The right side of formula (\ref{alx}) with $s_i=0$ can be written as $P(T_i=0)P(S_i=s_i, N=x|T_i=0)/P(S_i=s_i,N=x)$. Then, using again the equality $P(AB|C)=P(A|C)P(B|AC)$, we can represent $P(S_i=s_i, N=x|T_i=0)$ as a product $P(S_i=0|T_i=0)P(N=x|S_i=0,T_i=0)$. The first \prb in this product is $b$, the second \prb is $g_{n-1,k}(x-1)$. The denominator can be written as $P(S_i=s_i,N=x)=P(N=x)P(S=s|N=x)=P(N=x)x/n$. \par
Finally, cancelling common factors $P(T_i=0), P(N=x), g_{n-1,k}(x-1)$ and using the elementary equality $x\binom{n}{x}=n\binom{n-1}{x-1}$, we obtain that the left side of formula (\ref{alx}) coincides with the right side.
The proof for the case $s_i=1$ is similar.\par
\vspace{.2cm}     %%% and notation $d=1/c$
\emph{Proof of point c) of Lemma \ref{Lem2}. } We want to show that $r_A(x)\equiv r_{n,k}(x)$ depends on $a,b$ only through
$c=\frac{a}{1-a}\frac{b}{1-b}$. First, using formula (\ref{con}), we can represent $g_{n,k}(x)\equiv g(x)$ as
\vspace{.1cm}  %% g(x)=
\begin{eqnarray}
\sum_{i}\binom{k}{i}(1-a)^{i}a^{k-i}\binom{n-k}{x-i}b^{x-i}(1-b)^{n-k-x+i}=
a^{k}b^{x}(1-b)^{n-k-x}\sum_{i}\binom{k}{i}\binom{n-k}{x-i}c^{-i},\label{li1}
\end{eqnarray}
where summation over $i$ is carried with $d_1(x)=max(0,x-n+k)\leq i\leq min(k,x)=d_2(x)$. Then, using formula for $r_{n,k}(x)$  in Lemma 2, and formula (\ref{li1}) for $g_{n-1,k}(x-1)$ and $g_{n-1,k}(x)$ with \crsp \ summation, we can represent $r_{n,k}(x)$ as
\vspace{.1cm}
%%%%%%%%%%%%%% {d_1(x-1)\leq i\leq d_2(x-1)} {\sum_{d_1(x)\leq i\leq d_2(x)} %\large{
\begin{eqnarray}
r_{n,k}(x)=\frac{b}{(1-b)}\frac{(n-x)}{x}\frac{g_{n-1,k}(x-1)}{g_{n-1,k}(x)}=\frac{n-x}{x}
\frac{\sum_i\binom{k}{i}\binom{n-k-1}{x-i-1}c^{-i}}
{\sum_{i}\binom{k}{i}\binom{n-k-1}{x-i}c^{-i}}.\label{li2}
\end{eqnarray}
%} %\normalsize %%%%%%% or large {   } %%%%%%%%%%%%%%    POchemu - to LARGE font ?
Therefore $r _{n,k}(x)$ depends only on $c$. Since $c=\frac{ab}{1-(a+b)+ab}$, it is easy to see that $c>1$ iff $a+b>1$.
We also have $\binom{n-k-1}{x-i-1}= \binom{n-k-1}{x-i}\frac{(x-i)}{n-k-(x-i)}.$
Using these equalities and formula (\ref{li2}), we can show that $r_{n,k}(x)$ is decreasing in $c$ as a function of $c$.
\par %%%%%% eto 2      in LBT $B(n,k)$ Problem%%%%%The Proofs of points c), d), ?? f) of Lemma 1.}\label{Prf23}
%%%%%%% %%%%%%%%%%% two binom eq-s mvdAFTER
\vspace{.1cm}
We also need to prove that $r_{n,n-1}(x)=c$ for all $0<x<n$. Indeed, in this case the ratio
$g_{n-1,n-1}(x-1)/g_{n-1,n-1}(x)=\frac{a}{1-a}\frac{x}{n-x}$. This, combined with formula (\ref{rkx}) proves the claim.\par
\vspace{.2cm}
%%%%%%%%%%%%%%%%%% \label{STG} In our proof we are going to use three elementary equalities
Proof of Lemma \ref{Lem3}. We need to prove only (\ref{px}), since these equalities yield (\ref{rkx2}), and (\ref{stx}) follows from the definitions of $N_1$ and $N_2$. For the sake of brevity, let us denote event $F_t=\{N_1=t\}$. Note that $P(F_t|S=0,x)=P(F_t|S=1,x)=P(F_t|x)\equiv s(t|x)$. Indeed, $P(A|BC)=P(A|B)$ if $B\subseteq C$, and since $0<x<n$, an event $N=x$ is a subset of event $S=0,(S=1)$, that a particular box has a minus (plus). In other words, if after testing it is known that there are minus and plus boxes, then information that a particular box has minus or plus does not change the probability that there are $t$ locks among minuses. By symmetry of minus boxes, we also have $P(T=0|F_t,S=0,x)=\frac{x-t}{x}$. Therefore, by conditional total \prb formula, we have $P(T=0|S=0,x)=\sum_{t}P(F_t|S=0,x)P(T=0|F_t,S=0,x)=\sum_{t}s(t|x)\frac{x-t}{x}$.  Since $N_2=N-N_1$, $E(N_2|x)=x-E(N_1|x)$, we obtain the first formula in (\ref{px}). Similarly, using the equality $P(T=0|F_t,S=1,x)=\frac{n-k-(x-t)}{n-x}$, we obtain $P(T=0|S=1,x)=\sum_{t}s(t|x)\frac{n-k-(x-t)}{n-x}$, and since $U_2=n-k-N_2$, we obtain the second formula in (\ref{px}).\par
\vspace{.1cm}
\emph{Proof of Lemma \ref{Lem5}}. Suppose that Lemma \ref{Lem5}} is not true and let us say $u_r=i, u_t=j, i-j\geq 2$ and
$s_{r}=s_{t}=1.$ We remind that function $p(u)=1-q^u$ is concave. This concavity implies that $p(i-1)+p(j+1)>p(i)+p(j)$. Then, using the formulas in (\ref{pcalp}), we have
%%%%%%%%%%%%%%%%%%%%%%%% bylo &=&  \vspace{-.2cm}
\begin{eqnarray}
P(D&=&1|i-1,S=1,x)+P(D|j+1,S=1,x)=p^+(x)[p(i-1)+p(j+1)]> \notag \\
&>& p^+(x)[p(i)+p(j)]=P(D|i,S=1,x)+P(D|j,S=1,x). \label{P3}
\end{eqnarray}
Thus our initial \stg is not optimal. The proof for $s_{r}=s_{t}=0$ is similar with $p^+(x)$ replaced by $p^{-}(x)=r(x)p^+(x)$.\par
%%\vspace{.2cm}\par
\vspace{.1cm}
\emph{Proof of Lemma \ref{Lem6}}. As always, we assume that $a+b>1$. Then $r(x)>1$ for $0<x<n$, and hence $i\geq j$.
%We are using again shorthand notation $D(x)\equiv D$. Since $r(x)>1$, i.e., $P(T'|S'D)=r(x)p_+(x)>P(T'|SD)=p_+(x)$, then initially all bombs go sequentially into minus boxes, until all these boxes are filled.
As before we denote $P(\cdot |N=x)=P(\cdot |x)$, and denote the incremental utilities of destruction for minus and plus boxes as  $\Delta D^-(i|x)=P(D|i+1,S=0,x)-P(D|i,S=0,x), \ \Delta D^+(j|x)=P(D|j+1,S=1,x)-P(D|j,S=1,x)$, and the incremental utility of destruction with transition from the pair of values $(i,j)$ to the pair of values $(i-1,j+1)$ as $\Delta (i-1,j+1)$.
Using formulas in (\ref{pcalp}), it is easy to check that $\Delta D^+(j|x)=pp^+(x)q^j$ and
$\Delta D^-(i-1|x)=pr(x)p^+(x)q^{i-1}$, and then
%. If $i-j>d$, then we obtain that \par \vspace{.1cm} %% Then, formula (\ref{alp}) implies that their difference for $0\leq j\leq i$ is \vspace{-.2cm}
\begin{eqnarray}
\Delta (i-1,j+1)=\Delta D^-(i-1|x)-\Delta D^+(j|x)=pq^jp^+(x)[r(x)q^{i-j-1}-1].  \label{dc}
\end{eqnarray}
By the definition of $d(x)$, we have $r(x)q^{d(x)-1}\geq 1$ and $r(x)q^{d(x)}<1$. Then if $r(x)q^{d(x)-1}>1$ and $i-j>d=d(x)$, then formula (\ref{dc}) implies that $\Delta (i-1,j+1)<0$, i.e., a transfer of one bomb from a minus box from this pair to a plus box will increase the value of a \stge. Similarly, if $j\geq 1$ and $i-j<d-1$ for such pair, then using formula for $\Delta (i+1,j-1)$ similar to formula (\ref{dc}), we can show that the inverse transfer will increase the value.
If $r(x)q^{d(x)-1}=1$, then the incremental utility equals zero and two pairs $(i,j)$ and $(i-1,j+1)$ have the same destruction values. This implies that with such $d(x)$, $d(x)$-UAP \stg  remains optimal but not anymore unique.\par
%% \vspace{-.2cm}
Note also that, if $p=1$, i.e. $q=0$, then there is no need for more than $n$ bombs and first all minuses boxes are filled and after all plus boxes. If $p$ is decreasing to zero, then (\ref{dx}) shows that $d(x)$ tends to infinity.  \par
% \vspace{.2cm} \vspace{-.2cm}
\subsection{ Short Remarks on Case $B(n,\lambda,m)$ }\label{Bn}
\vspace{.1cm} %fr Lemma 2 was point b)
A theorem similar to Theorem \ref{th:1} was proved for Problem $B=B(n,\lambda,m)$ in \cite{lison21}.  The value of a threshold $d_B$ is different from $d_A(x)$ and its description is more complicated involving the aposterior \dst of locks after testing. Although in problem $B$ the ratio $r_B$ does not depend on $x$, in \emph{both} problems the optimal allocation of $m$ bombs depends on $x$.  A similar situation holds for their optimal value functions. The \flw statement (see \cite{lison21}) is  true for Problem $B$. \par
\vspace{.1cm}
The ratio $r_B$ has a form
%%%%%%%%%%%%%%theorem \vspace{-.3cm}
\begin{eqnarray}
			 r_B=\frac{p^{-}}{p^{+}}=\frac{b}{1-b}\frac{1-b+\lambda h}{b-\lambda h}, \ \ 0<\lambda<1, \label{rB}
		\end{eqnarray}
%\vspace{-.2cm}
where $h=a+b-1$.
%%%%%%%%%%% Also of possible interest is also to compare the strategies and the value functions for both problems for the "matching" values of $k$ and $\lambda$, i.e. when the \emph{expected number} of locks in $n$ boxes is the same, $\lambda n=k$.
%\newpage \section{Compliance with Ethical Standards} %
\section{Statements and Declarations} \begin{itemize}
\item No funds, grants, or other support was received.
\item  Ethical approval: This article does not contain any studies with human participants or animals performed by any of the authors.
\item The authors have no relevant financial or non-financial interests to disclose.
\end{itemize}

\end{document}